\documentclass{LMCS}

\usepackage{proof}
\usepackage{latexsym}
\usepackage{amssymb}
\usepackage{amsmath}

\usepackage{url}
\usepackage{rotating}
\usepackage[all]{xy}
\usepackage{pstricks}
\usepackage{pst-node}
\usepackage{pst-tree}
\usepackage{enumerate,hyperref}

\def\impl {\to}
\def\Kt{\mathit{Kt}}
\def\myK{\hbox{K}}
\def\DKt{\mathrm{\bf DKt}}
\def\SKt{\mathrm{\bf SKt}}
\def\S2Kt{\mathrm{\bf S2Kt}}
\def\DK{\mathrm{\bf DK}}
\def\DKFour{\mathrm{\bf DKS4}}
\def\DKFive{\mathrm{\bf DKS5}}
\def\SSU{\mathrm{\bf SSU}}
\def\DKtU{\mathrm{\bf DKtU}}
\def\DKU{\mathrm{\bf DKU}}
\newcommand\SKtExt[1]{\mathrm{\bf SKt #1}}
\newcommand\DKtExt[1]{\mathrm{\bf DKt #1}}
\def\DKP{\mathrm{DKP}}

\def\Scal{\mathcal{S}}

\def\SSFour{\mathrm{\bf SS4}}
\def\DSFour{\mathrm{\bf DS4}}

\def\SSFive{\mathrm{\bf SS5}}
\def\DSFive{\mathrm{\bf DS5}}

\def\udia{\mbox{$\langle ?  \rangle$}}
\def\ubox{\mbox{$[ ? ]$}}
\newcommand\seqdia[1]{\mbox{$\langle #1 \rangle$}}

\def\prod{\twoheadrightarrow}

\newcommand{\FBox}{\square}
\newcommand{\FDia}{\lozenge}
\newcommand{\PBox}{\blacksquare}
\newcommand{\PDia}{\blacklozenge}
\newcommand{\pseq}[1]{\bullet\{ #1 \}}
\newcommand{\fseq}[1]{\circ\{ #1 \}}
\newcommand{\lneg}[1]{\overline #1}

\newcommand{\pseqn}[2]{{\bullet^{#1} \{ #2 \} }}
\newcommand{\fseqn}[2]{{\circ^{#1} \{ #2 \} }}

\newcommand{\mbf}[1]{\mathrm{\bf #1}}

\def\Ccal{\mathcal{C}}

\newenvironment{definition}{\begin{defi}}{\end{defi}}
\newenvironment{proposition}{\begin{prop}}{\end{prop}}
\newenvironment{lemma}{\begin{lem}}{\end{lem}}
\newenvironment{theorem}{\begin{thm}}{\end{thm}}
\newenvironment{corollary}{\begin{cor}}{\end{cor}} 
\newenvironment{example}{\begin{exa}}{\end{exa}} 

\def\doi{7 (2:8) 2011}
\lmcsheading%
{\doi}
{1--38}
{}
{}
{Jun.~11, 2010}
{May\phantom.~16, 2011}
{}

\begin{document}

\title[On the Correspondence between Display Postulates and Deep Inference]
{On the Correspondence between Display Postulates and Deep Inference in
Nested Sequent Calculi for Tense Logics\rsuper*}
\author[R.~Gor\'e]{Rajeev Gor\'e} 
\author[L.~Postniece]{Linda Postniece}
\author[A.~Tiu]{Alwen Tiu}
\address{
Logic and Computation Group \\ 
College of Engineering and Computer Science \\
The Australian National University 
}
\email{\{Rajeev.Gore,Alwen.Tiu\}@anu.edu.au, linda.postniece@gmail.com}


\titlecomment{{\lsuper*}This article is a revised and extended version of 
an extended abstract presented 
at the TABLEAUX 2009 conference~\cite{Gore09tableaux}.}

\begin{abstract}
We consider two styles of proof calculi for a family of tense logics,
presented in a formalism based on nested sequents. A nested sequent
can be seen as a tree of traditional single-sided sequents. Our first
style of calculi is what we call ``shallow calculi'', where inference
rules are only applied at the root node in a nested sequent. Our
shallow calculi are extensions of Kashima's calculus for tense logic
and share an essential characteristic with display calculi, namely,
the presence of structural rules called ``display postulates''.
Shallow calculi enjoy a simple cut elimination procedure, but are
unsuitable for proof search due to the presence of display postulates
and other structural rules. The second style of calculi uses
deep-inference, whereby inference rules can be applied at any node in
a nested sequent. We show that, for a range of extensions of tense
logic, the two styles of calculi are equivalent, and there is a
natural proof theoretic correspondence between display postulates and
deep inference. The deep inference calculi enjoy the subformula
property and have no display postulates or other structural rules,
making them a better framework for proof search.
\end{abstract}

\maketitle

\section{Introduction}
\label{sec:intro}

A nested sequent is essentially a tree whose nodes are traditional
sequents. It has been used as the syntactic
judgment for proof calculi for several tense and modal 
logics~\cite{kashima-cut-free-tense,brunnler2006,poggiolesi2009,BruHabil}, 
perhaps due to the fact that the tree structure embodies, to some extent, the
underlying Kripke frames in those logics. 
In our setting,  the nodes in a nested sequent are traditional 
single-sided sequents (i.e., multisets of formulae), and the edges connecting 
the nodes are labelled either with $\circ$ or $\bullet$ (these labels correspond
to the modal operator $\FBox$ and the tense operator $\PBox$). 
For example, the trees shown in Figure~\ref{fig:tree} are a 
tree representation of nested sequents, where each $\Gamma_i$ is a multiset 
of formulae.

\begin{figure}[t]
$$
\begin{array}{cccc}
(a) & 
\mbox{
\pstree[levelsep=25pt,labelsep=0pt]{\TR{$\Gamma_1$}}
{
\TR{$\Gamma_2$}\tlput{$\circ$}
\pstree{\TR{$\Gamma_3$}\trput{$\circ$}}
{
 \TR{$\Gamma_4$}\tlput{$\bullet$}
 \TR{$\Gamma_5$}\trput{$\circ$}
} 
}
}
&
\raisebox{-5ex}{$\leadsto$}
&
\mbox{
\pstree[levelsep=30pt,labelsep=0pt]{\TR{$\Gamma_3$}}
{
\pstree{\TR{$\Gamma_1$}\tlput{$\bullet$}}
       {\TR{$\Gamma_2$}\tlput{$\circ$}}
\TR{$\Gamma_4$}\tlput{$\bullet$}
\TR{$\Gamma_5$}\trput{$\circ$}
}
}
\\
(b)
&
\qquad
\mbox{
\pstree[levelsep=30pt,labelsep=0pt]{\TR{$\Gamma_1$}}
{
\TR{$\Gamma_2$}\tlput{$\circ$}
\TR{$\Gamma_3$}\trput{$\circ$}
}
}
&
\raisebox{-5ex}{$\leadsto$}
&
\mbox{
\pstree[levelsep=30pt,labelsep=0pt]{\TR{$\Gamma_1$}}
{
\TR{$\Gamma_2$}\tlput{$\circ$}
\pstree{\TR{ $\emptyset$ }\trput{$\circ$}}
{
\TR{$\Gamma_3$}\trput{$\circ$}
}
}
}  \\
(c) & 
\mbox{
\pstree[levelsep=30pt,labelsep=0pt]{\TR{$\Gamma_1$}}
{
\TR{$\Gamma_2$}\tlput{$\circ$}
\pstree{\TR{$\Gamma_3$}\trput{$\circ$}}
{
 \TR{$\Gamma_4, A$}\tlput{$\circ$}
 \TR{$\Gamma_5$}\trput{$\bullet$}
} 
}
}
&
\raisebox{-5ex}{$\leadsto$}
&
\mbox{
\pstree[levelsep=30pt,labelsep=0pt]{\TR{$\Gamma_1$}}
{
\TR{$\Gamma_2$}\tlput{$\circ$}
\pstree{\TR{$\Gamma_3$}\trput{$\circ$}}
{
 \TR{$\Gamma_4$}\tlput{$\circ$}
 \TR{$\Gamma_5, \FDia A$}\trput{$\bullet$}
} 
}
}
\end{array}
$$
\caption{Inference rules seen as operations on trees}
\label{fig:tree}
\end{figure}

There are two natural styles of formalising inference rules on 
nested sequents. The first is one that conforms with the tradition of
sequent calculi, namely, to allow inference rules to act only on formulae
or structures that appear at the root sequent. We shall refer to this
style of inference as {\em shallow inference}. The second style is
to allow inference rules to act on formulae or structures in an arbitrary
node in the tree; we call this {\em deep inference}. 
Kashima's work~\cite{kashima-cut-free-tense} includes 
inference systems of both kinds. More specifically, Kashima
presents two proof systems for tense logic, a shallow proof system $\SKt$
and a deep-inference system $\S2Kt$, and proves, via semantical
methods, that they are equivalent. 
In this paper, we investigate, via proof theoretic methods, the connection 
between shallow and deep inference systems for a wide range of tense logics 
extending Kashima's $\SKt$ and $\S2Kt.$ 

The primary motivation of our work actually stems from the problem
of structuring proof search for display calculi~\cite{Belnap82JPL};
more specifically Kracht's formulation of display calculi for 
extensions of tense logics~\cite{kracht-power}. We have yet to tackle the
proof search problem for Kracht's display calculi in their full generality.
What we show here is that in a more restricted setting of nested sequent 
calculi, which can be seen as a restricted form of display calculi, 
one main impediment to proof search, i.e., unrestricted use
of structural rules, can be eliminated. In particular, we aim for
a uniform design methodology for deep inference calculi {\em without
structural rules}. This design choice sets us apart from similar work 
by Br\"unnler and Stra{\ss}burger~\cite{Brunnler09tableaux}, 
where structural rules in the deep inference systems are actually desirable,
out of the consideration for modularity 
(see also the discussion in Section~\ref{sec:related}).

\paragraph{\em Display postulates and other structural rules}
Kashima's shallow calculus $\SKt$ shares an essential feature with
Kracht's display calculi, namely, the presence of the so-called
{\em display postulates} (called the {\em turn} rules 
in \cite{kashima-cut-free-tense}). 
Seen as an operation on trees, the display postulates are a rotation 
operation on trees, allowing one to bring an arbitrary node
in a tree to the root, e.g., the transformation 
shown in Figure~\ref{fig:tree}(a) ``displays'' the sequent $\Gamma_3.$ 

An interesting result in Kracht's work \cite{kracht-power} is that one can 
construct display calculi for extensions of tense logic {\em modularly}. 
That is, for every axiom in a certain form called {\em primitive form}, 
one can construct a structural rule that captures exactly that axiom. 
Due to the similarity between display calculi and our shallow calculi,
Kracht's approach can be adapted to our setting as well to design modular
shallow calculi. 
Seen as operations on trees, structural rules induced by axioms may involve 
addition or removal of nodes in the trees, e.g., the transitivity axiom,
$\FBox A \impl \FBox \FBox A$, translates to the operation 
shown in Figure~\ref{fig:tree}(b). In addition to these structural rules,
our shallow calculi (and Kracht's display calculi) also contain contraction
and weakening rules, which allow duplication and removal of arbitrary subtrees. 
A combination of all these structural rules presents a complication
in using display calculi or shallow nested-sequent calculi as a framework
to structure proof search. 

\paragraph{\em Deep inference and propagation rules.}
The role of display postulates is really to move a sequent to 
the root of a nested sequent so that an inference rule
may be applied to it. Therefore a natural way to eliminate display postulates
is to just allow inference rules to be applied deeply, as already shown by 
Kashima in his proof of the correspondence between $\SKt$ and 
$\S2Kt$~\cite{kashima-cut-free-tense}. 
However, for extensions of tense logics, deep inference alone is not enough.
For example, in the extension with the transitivity axiom, the problem
is not so much that one cannot apply rules deeply. Rather, it is more to do with
the fact that we extend the tree of sequents with extra nodes. 
To eliminate the structural rules for transitivity, we
need to somehow build in transitivity into {\em logical rules}. 
We do this systematically via the so-called {\em propagation rules}. 
More specifically, the introduction rules for $\FDia$-formulae (and its tense 
counterparts), reading the rules bottom up, allow propagations of the formulae 
along certain paths in the nested sequent. As an illustration,
consider the instance of a propagation rule needed to absorb the transitivity 
axiom given in Figure~\ref{fig:tree}(c), where a formula $A$ in one node 
(where $\Gamma_4$ resides) is propagated to another node (where $\Gamma_5$ 
resides). We defer the justification for this rule to Section~\ref{sec:path}; 
for now, we just note that one can introduce a $\FDia$-prefixed formula across 
different nodes at arbitrary depth in the tree, not just the top node.

\paragraph{\em Summary of results.} Our main contributions are the following:
\begin{enumerate}[$\bullet$]
\item We give a uniform syntactic cut elimination procedure for 
extensions of $\SKt$ with what we call {\em linear structural rules} 
(Section~\ref{sec:skt}). 
Our procedure is very similar to 
Belnap's general cut elimination for display logics, as it relies on the
existence of the ``display property'' for our shallow calculi. It can be seen
as an adaptation of Kracht's cut elimination for display calculi for tense 
logics~\cite{kracht-power} to the setting of nested sequent calculi.  
Existing works on syntactic cut elimination for nested sequent calculi 
address only the modal fragment (in the deep inference setting) and 
for a limited number of extensions, e.g., 
\cite{brunnler2006,Brunnler09AML,Brunnler09tableaux,poggiolesi2009}, 
or only for some extensions of tense logic without negation or 
implication~\cite{Sadrzadeh10}.

\item We show that for two classes of axioms, the Scott-Lemmon 
axioms~\cite{Lemmon77} and {\em path axioms} (Section~\ref{sec:path}), 
the axioms can be modularly turned into
linear structural rules, and hence cut admissibility for the shallow systems
for these extensions follows from our uniform cut elimination. 
These two classes of axioms cover most of standard normal modal axioms in 
the literature, e.g., reflexivity, transitivity,
euclideanness, convergence, seriality, etc. 

\item We give a syntactic proof of the equivalence of $\SKt$ and $\S2Kt$ (which
we call $\DKt$ here). Kashima gave a proof of this correspondence via a 
semantic argument~\cite{kashima-cut-free-tense}. 
We further show that, for some extensions of $\SKt$ with
Scott-Lemmon axioms, one can get the corresponding deep inference systems 
extending $\DKt$, {\em without structural rules}, but with {\em local} 
propagation rules (Section~\ref{sec:ext}). 
For the extensions with path axioms, we show how one can derive systematically
the corresponding deep inference calculi, also without structural rules,
but with {\em global} propagation rules. 
By local propagation rules, we mean propagation rules in which formulae
may be propagated only along a path of bounded length, whereas global
propagation rules do not restrict the length of the path.

\item We show that all our deep inference calculi for tense logics enjoy 
the {\em separation property}:
if one restricts the calculi to their modal fragments, i.e., by omitting rules
that mention tense operators, then one gets complete calculi for the modal parts
of the tense logics. 
\end{enumerate}
The relationships between various proof systems in this paper are
summarised in Figure~\ref{fig:relationship}. The direction of an arrow denotes 
inclusion, e.g., Kashima's $\S2Kt$ is equivalent to $\DKt$. The dashed arrow
in the lowest row denotes the fact that the equivalence is only established
for some, but not all, Scott-Lemmon axioms. We have not yet explored the
connections between path axioms and Scott-Lemmon axioms.  

\begin{figure}[t]
$$
\xymatrix{
 & \mbox{
   $\begin{array}{c} 
     \SKt ~ + \\ 
     \mbox{path axioms}
     \end{array}$
   } 
  \ar@{<->}[r]
  & \mbox{
    $\begin{array}{c}
      \DKt + \\
      \mbox{global} \\ 
      \mbox{propagation rules}
     \end{array}
    $} \\
 \mbox{
   $
   \begin{array}{c}
     \mbox{Kashima's} \\
     \SKt
   \end{array}
   $
  } \ar@{<->}[r] & 
  \SKt \ar@{->}[u] \ar@{->}[d] \ar@{<->}[r]
  & \DKt \ar@{->}[u] \ar@{->}[d] \ar@{<->}[r]
  & \mbox{
     $
     \begin{array}{c}
      \mbox{Kashima's} \\ 
      \S2Kt
     \end{array}
     $
    } \\
& \mbox{
  $\begin{array}{c}
   \SKt ~ + \\ 
   \mbox{Scott-Lemmon} \\
   \mbox{axioms}
   \end{array}
  $ } \ar@{<.>}[r] & 
 \mbox{
  $\begin{array}{c}
  \DKt ~ + \\
  \mbox{local} \\
  \mbox{propagation rules}
  \end{array}
  $ }
}
$$
\caption{Relationships between proof systems}
\label{fig:relationship}
\end{figure}

\paragraph{\em Outline of the paper}
Section~\ref{sec:logic} gives an overview of the syntax and the semantics of
tense logic. Section~\ref{sec:skt} presents the shallow calculus $\SKt$ and 
a uniform syntactic cut elimination proof for any extension of $\SKt$ with linear structural rules. 
Section~\ref{sec:dkt} presents a deep inference calculus $\DKt$, which is similar 
to Kashima's $\S2Kt$, but without structural rules. We prove that
$\SKt$ and $\DKt$ are cut-free equivalent, i.e., any cut-free proof in $\SKt$
can be transformed into a cut-free proof in $\DKt$ and vice versa. 
Section~\ref{sec:ext} presents extensions of $\SKt$ with Scott-Lemmon axioms.
We show that for some extensions, one can 
design deep inference calculi based on $\DKt$ extended with some local
propagation rules. Section~\ref{sec:path} considers extensions of $\SKt$
with path axioms. We show how these axioms can be captured using global
propagation rules in deep inference. We show further
that applicability of propagation rules is decidable, by mapping the
decision problem into the problem of non-emptiness checking of the intersection
of a context-free language and a regular language. 
Section~\ref{sec:search} gives some preliminary results in proof search for
$\DKt$. Section~\ref{sec:related} concludes
the paper and discusses related and future work. 

This paper is a revised and extended version of an extended abstract presented
at the TABLEAUX 2009 conference~\cite{Gore09tableaux}. 
We have added the following new material: 
a uniform cut elimination proof for extensions of $\SKt$ 
with linear structural rules, extensions of $\SKt$ with Scott-Lemmon axioms, 
a new extension of $\DKt$, and a new section (Section~\ref{sec:path}) on path axioms.
However, we have removed the material on proof search for $KtS4$ in the conference
version, as we have recently discovered that the proof search algorithm outlined
in that paper is unsound, although the calculi themselves are sound
and complete. We defer the complete treatment of proof
search for extensions of $\DKt$ to future work.

\section{Tense Logic}
\label{sec:logic}

To simplify presentation, we shall consider formulae of tense logic
$\Kt$ which are in negation normal form (nnf), given by the following grammar:
$$
A := a 
		\mid \lnot a
		\mid A \lor A
		\mid A \land A
		\mid \FBox A
		\mid \PBox A
		\mid \FDia A
		\mid \PDia A
$$
where $a$ ranges over atomic formulae and $\lnot a$ is the negation of $a$. 
We shall denote with $\lneg A$ the nnf of the negation of $A$. 
Implication can then be defined via negation: $A \impl B = \lneg A \lor B.$
The axioms of minimal tense logic $\Kt$ 
are all the axioms of propositional logic, plus the axioms
in Figure~\ref{fig:tense-axioms}.
\begin{figure}[t]
$$
\begin{array}{ll@{\qquad}l}
(1) & A \impl \FBox \PDia A & \lneg A \vee \FBox \PDia A \\
(2) & A \impl \PBox \FDia A & \lneg A \vee \PBox \FDia A\\
(3) & \FBox(A \impl B) \impl (\FBox A \impl \FBox B) & 
\FDia(A \land \lneg B) \lor \FDia \lneg A \lor \FBox B\\
(4) & \PBox(A \impl B) \impl (\PBox A \impl \PBox B) & 
\PDia(A \land \lneg B) \lor \PDia \lneg A \lor \PBox B.
\end{array}
$$
\caption{Axioms of minimal tense logic. Their nnf are shown on the right hand side.}
\label{fig:tense-axioms}
\end{figure}

The theorems of $\Kt$ are those that are generated from the above
axioms and their substitution instances using the following
rules:
$$
\infer[MP]
{B}
{A & \lneg A \lor B}
\qquad
\infer[Nec\FBox]
{\FBox A}
{A}
\qquad
\infer[Nec\PBox]
{\PBox A}
{A}
$$

A {\em $\Kt$-frame} is a pair $\langle W, R \rangle$, with $W$ a
non-empty set (of worlds) and $R \subseteq W \times W$. A $\Kt$-model
is a triple $\langle W, R, V \rangle$, with $\langle W, R \rangle$ 
a $\Kt$ frame and $V: Atm \to 2^W$ a valuation mapping each atom
to the set of worlds where it is true.

For a world $w \in W$ and an atom $a \in Atm$, if $w \in V(a)$ then we
write $w \Vdash a$ and say $a$ is forced at $w$; otherwise we write $w
\not\Vdash a$ and say $a$ is rejected at $w$. Forcing and rejection of
compound formulae is defined by mutual recursion in
Figure~\ref{fig:forcing}. A $\Kt$-formula $A$ is valid iff it is
forced by all worlds in all models, i.e. iff $w \Vdash A$ for all
$\langle W, R, V \rangle$ and for all $w \in W$.

\begin{figure}[t]
$$
\begin{array}{llllll}
w \Vdash \neg A		&	 \text{iff}	 &	w \not\Vdash A \ \ & \ \ 
\\
w \Vdash A \lor B	&	 \text{iff}	 &		w \Vdash A \text{ or } w \Vdash B \ \ & \ \
w \Vdash A \land B	&	 \text{iff}	 &		w \Vdash A \text{ and } w \Vdash B \\
w \Vdash \FBox A	&	 \text{iff}	 &		\forall u. \text{ if } w R u \text { then } u \Vdash A \ \ & \ \	
w \Vdash \FDia A	&	 \text{iff}	 &		\exists u. w R u \text { and } u \Vdash A \\	
w \Vdash \PBox A	&	 \text{iff}	 &		\forall u. \text{ if } u R w \text { then } u \Vdash A \ \ & \ \	
w \Vdash \PDia A	&	 \text{iff}	 &		\exists u. u R w \text { and } u \Vdash A \\	
\end{array}
$$
\caption{Forcing of formulae}
\label{fig:forcing}
\end{figure}

\section{System \texorpdfstring{$\SKt$}{SKt}: a ``shallow'' calculus}
\label{sec:skt}

We consider a right-sided proof system for tense logic where
the syntactic judgment is a tree of multisets of formulae,
called a {\em nested sequent}. Nested sequents have been used 
previously in proof systems for modal and tense 
logics~\cite{kashima-cut-free-tense,brunnler2006,poggiolesi2009}.

\begin{definition}
A {\em nested sequent} is a multiset
$$\{A_1, \ldots , A_k, \fseq {\Gamma_1}, \ldots, \fseq{\Gamma_m}, 
 \pseq{\Delta_1}, ..., \pseq{\Delta_n}\}
$$
where $k, m, n \geq 0$, and each $\Gamma_i$ and each $\Delta_j$ are themselves
nested sequents.
\end{definition}
We shall use the following notational conventions when writing
nested sequents. 
We shall remove outermost braces, e.g., we write $A, B, C$
instead of $\{A,B,C\}.$ 
Braces for sequents nested inside $\fseq{}$ or $\pseq{}$ are also removed, 
e.g., instead of writing $\fseq{\{A,B,C\}}$, 
we write $\fseq{A,B,C}$. The empty (nested) sequent is denoted by $\emptyset.$
When we juxtapose two sequents as in  $\Gamma, \Delta$ 
we mean a sequent resulting from the multiset-union of $\Gamma$ and 
$\Delta$. When $\Delta$ is a singleton multiset, e.g., $\{A\}$ or 
$\{\fseq{\Delta'}\}$, we simply write: $\Gamma, A$ or $\Gamma, \fseq{\Delta'}.$
Since we shall only be concerned with nested sequents,
we shall refer to nested sequents simply as sequents in the rest of
the paper.  

The above definition of sequents can also be seen as a special
case of {\em structures} in display calculi, e.g., with `,' (comma),
$\bullet$ and $\circ$ as structural connectives~\cite{Gore98IGPL}. 

A {\em context} is a sequent with holes in place of  
formulae. A context with a single hole is written as $\Sigma[]$.
Multiple-hole contexts are written as $\Sigma[]\cdots[]$,
or abbreviated as $\Sigma^k[]$ where $k$ is the number
of holes. We write $\Sigma^k[\Delta]$ to denote the sequent
that results from filling the holes in $\Sigma^k[]$ uniformly with
$\Delta.$

Given a proof system $\mbf{S}$, a {\em derivation in $\mbf{S}$}
is defined as usual, i.e., as a tree whose nodes are nested sequents
such that every node is the conclusion of an inference rule in $\mbf{S}$,
and all its child nodes are exactly the premises of the same rule. An {\em open
derivation in $\mbf{S}$} may additionally contain one or more leaf nodes,
called {\em open leaf nodes}, which are not conclusions of any rules in $\mbf{S}.$
We say that a sequent $\Gamma$ is {\em derivable from $\Delta$ in $\mbf{S}$}
if there is an open derivation of $\Gamma$ whose open leaf nodes
are $\Delta.$

The shallow proof system for $\Kt$, called $\SKt$, is given 
in Figure~\ref{fig:SKt}. Note that the $id$-rule is restricted to the
atomic form, but it is easy to show that the general $id$ rule on
arbitrary formulae is admissible.
$\SKt$ is basically Kashima's system for tense logic 
(also called $\SKt$)~\cite{kashima-cut-free-tense},
but with a more general contraction rule ($ctr$), which allows
contraction of arbitrary sequents. 
The general contraction rule is used to simplify our cut elimination proof,
and as we shall see in Section~\ref{sec:dkt}, 
it can be replaced by formula contraction.
System $\SKt$ can also be seen as a single-sided version of a display calculus.  
The rules $rp$ and $\mathit{rf}$ are called the {\em residuation rules} \cite{Gore98IGPL}.  
They are an example of {\em display postulates} commonly found in display
calculi, and are used to bring a node in a nested sequent to the top
level. The following is an analog of the display 
property of display calculus. Its proof is straightforward by induction
on the size of contexts.  

\begin{proposition}
\label{prop:display}
Let $\Sigma[\Delta]$ be a sequent. Then there exists a 
sequent $\Gamma$ such that $\Sigma[\Delta]$  is derivable 
from $\Delta, \Gamma$ and vice versa, using only the rules $\mathit{rp}$ and $\mathit{rf}$.
\end{proposition}

\begin{figure}[t]
$$
\begin{array}{c@{\qquad \quad}c@{\qquad \quad}c@{\qquad \quad}c}
\infer[\mathit{id}]
{\Gamma, a, \bar a}{}
&
\infer[\mathit{cut}]
{\Gamma, \Delta}
{\Gamma, A & \Delta, \lneg A}
&
\infer[\land]
{\Gamma, A \land B}
{\Gamma, A & \Gamma, B}
&
\infer[\lor]
{\Gamma, A \lor B}
{\Gamma, A, B}
\\ \\
\infer[\mathit{ctr}]
{\Gamma, \Delta}
{\Gamma, \Delta, \Delta}
&
\infer[\mathit{wk}]
{\Gamma, \Delta}
{\Gamma}
&
\infer[\mathit{rf}]
{\pseq \Gamma, \Delta}
{\Gamma, \fseq \Delta}
&
\infer[\mathit{rp}]
{\fseq \Gamma, \Delta}
{\Gamma, \pseq \Delta}
\\ \\
\infer[\PBox]
{\Gamma, \PBox A}
{\Gamma, \pseq A}
&
\infer[\FBox]
{\Gamma, \FBox A}
{\Gamma, \fseq A}
&
\infer[\PDia]
{\Gamma, \pseq \Delta, \PDia A}
{\Gamma, \pseq {\Delta, A}}
&
\infer[\FDia]
{\Gamma, \fseq \Delta, \FDia A}
{\Gamma, \fseq{\Delta, A}}
\end{array}
$$
\caption{System $\SKt$}
\label{fig:SKt}
\end{figure}

\subsection{Soundness and completeness}

To prove soundness, we first show that each 
sequent has a corresponding $\Kt$-formula, and then
show that the rules of $\SKt$, reading them top down, 
preserve validity of the formula corresponding to the premise sequent. 
Completeness is shown by simulating the Hilbert system for tense logic in
$\SKt.$
The translation from sequents to formulae are given below. 
In the translation, we assume two logical constants $\bot$ (`false') and $\top$ (`true'). 
This is just a notational convenience, as the constants can be defined in a standard way,
e.g., as $a \land \bar a$ and $a \lor \bar a$ for some fixed atomic proposition $a$. 
As usual, the empty disjunction denotes $\bot$ and the empty conjunction denotes $\top.$

\begin{definition}
  The function $\tau$ translates an $\SKt$-sequent 
$$\{A_1, \ldots , A_k, \fseq {\Gamma_1}, \ldots, \fseq{\Gamma_m}, 
 \pseq{\Delta_1}, ..., \pseq{\Delta_n}\}
$$
into the $\Kt$-formula (modulo associativity and commutativity of
$\lor$ and $\land$):
$$
A_1 \lor \cdots \lor A_k \lor 
\FBox \tau(\Gamma_1) \lor \cdots \lor \FBox \tau(\Gamma_m) \lor
\PBox \tau(\Delta_1) \lor \cdots \lor \PBox \tau(\Delta_n).
$$
\end{definition}

\begin{lemma}[Soundness]
\label{lm:SKt-soundness}
Every $\SKt$-derivable $\Kt$ formula is valid.
\end{lemma}
\begin{proof}
We show that for every rule $\rho$ of $\SKt$ 
$$ 
\infer[\rho] {\Gamma}
{\Gamma_1 & \cdots & \Gamma_n} 
$$ 
the following holds: if for every $i \in \{1,\ldots , n\}$, 
the formula $\tau(\Gamma_i)$ is valid then the
formula $\tau(\Gamma)$ is valid. 

Since the formula-translation 
$\tau(\Gamma) \lor a \lor \lneg a$ of the $id$ rule is obviously valid, 
it then follows that every formula derivable in $\SKt$ is also valid. 
We show the soundness of $\mathit{rf}$ here; the others are similar or easier:
We want to show that if $\tau(\Gamma) \lor \FBox(\tau(\Delta))$ is valid 
then $\PBox(\tau(\Gamma)) \lor \tau(\Delta)$ is valid. We prove this by contradiction.
Suppose $\tau(\Gamma) \lor \FBox(\tau(\Delta))$ is valid but
$\PBox(\tau(\Gamma)) \lor \tau(\Delta)$ is not, so there is a model 
$\langle W, R, V\rangle$ and a world $w \in W$ 
such that $w \not \models \PBox(\tau(\Gamma)) \lor \tau(\Delta)$, which means
\begin{equation}
\label{eq:soundness}
w \not \models \PBox(\tau(\Gamma)) \hbox{ and }
w \not \models \tau(\Delta).
\end{equation}
Since $w \not \models \PBox(\tau(\Gamma))$, there must be a world $v \in W$
such that $v R w$ and $v \not \models \tau(\Gamma).$
Now since $\tau(\Gamma) \lor \FBox(\tau(\Delta))$ is valid, we have
$v \models \tau(\Gamma)$ or $v \models \FBox(\tau(\Delta)).$
But because $v \not \models \tau(\Gamma)$, it follows that
$v \models \FBox(\tau(\Delta))$. Since $v R w$, by definition,
we have $w \models \tau(\Delta)$, which contradicts our assumption
above in (\ref{eq:soundness}). 
\end{proof}

\begin{lemma}[Completeness]
\label{lm:SKt-completeness}
Every $\Kt$-theorem is $\SKt$-derivable.
\end{lemma}
\begin{proof}
The proof follows a standard translation from Hilbert systems 
to Gentzen's systems (see, e.g., \cite{Troelstra96bpt}). 
We show here only derivations of Axioms (1) and (3) in Figure~\ref{fig:tense-axioms};
the other axioms and rules are not difficult to handle. 
Double lines abbreviate derivations:
$$
\infer[\lor]
{\lneg A \vee \FBox \PDia A}
{
\infer[\FBox]
{\lneg A, \FBox \PDia A}
{
\infer[\mathit{rp}]
{\lneg A, \fseq {\PDia A}}
{
\infer[\PDia]
{\pseq {\lneg A}, \PDia A}
{
\infer[\mathit{rf}]
{\pseq {\lneg A, A}}
{
\infer[\mathit{id}]
{\fseq{\ }, \lneg A, A}
{}
}
}
}
}
}
\qquad \qquad
\infer=[\lor]
{\FDia(A \land \lneg B) \lor \FDia \lneg A \lor \FBox B}
{
\infer[\FBox]
{\FDia(A \land \lneg B),\FDia \lneg A,\FBox B}
{
\infer=[\FDia]
{\FDia(A \land \lneg B),\FDia \lneg A,\fseq B}
{
\infer[\mathit{rp}]
{\fseq {A \land \lneg B, \lneg A, B}}
{
\infer[\land]
{A \land \lneg B, \lneg A, B, \pseq{\ }}
{
\infer[\mathit{id}]
{A, \lneg A, B, \pseq{\ }}
{}
&
\infer[\mathit{id}]
{\lneg B, \lneg A, B, \pseq{\ }}
{}
}
}
}
}
}
$$
\end{proof}

The following theorem is then a simple corollary of the
Lemma~\ref{lm:SKt-soundness} and Lemma~\ref{lm:SKt-completeness}.

\begin{theorem}
\label{thm:soundness-completeness}
A $\Kt$-formula $A$ is valid iff
$A$ is $\SKt$-derivable.
\end{theorem}

\subsection{Cut elimination}

\begin{figure}[t]
$$
\begin{array}{ccc}
\infer[\mathit{cut}]
{\fseq \Gamma, \pseq \Delta}
{
 \infer[\mathit{rf}]
 {\fseq \Gamma, A}
 {\deduce{\Gamma, \pseq A}{\Pi_1}}
 & 
 \infer[\mathit{rp}]
 {\lneg A, \pseq \Delta }
 {
  \deduce{\fseq {\lneg A}, \Delta }{\Pi_2}
 }
}
&
\quad
\deduce{\deduce{\Gamma, \pseq {A_1 \land A_2}}{\vdots}}
{
 \infer[\mathit{rf}]
 {\Gamma', \pseq {A_1 \land A_2}}
 {
   \infer[\land]
   {\fseq {\Gamma'}, A_1 \land A_2}
   {
    \deduce{\fseq {\Gamma'}, A_1}{\vdots} & \deduce{\fseq {\Gamma'}, A_2}{\vdots}
   }
 }
}
&
\quad
\deduce{\deduce{ \fseq {\lneg{A_1} \lor \lneg{A_2}}, \Delta}{\vdots}}
{
 \infer[\mathit{rp}]
 {\fseq {\lneg{A_1} \lor \lneg{A_2}}, \Delta'}
 {
   \infer[\lor]
   {\lneg {A_1} \lor \lneg {A_2}, \pseq {\Delta'}}
   {
     \deduce{\lneg {A_1} , \lneg {A_2}, \pseq {\Delta'}}{\vdots}
   }
 }
} \\ \\
(1) & \quad (2) & \quad (3)
\end{array}
$$
$$
\begin{array}{cc}
\infer[\mathit{rp}]
{\fseq \Gamma, \pseq \Delta}
{
 \deduce{\deduce{\Gamma, \pseq {\pseq \Delta}}{\vdots}}
 {
  \infer[\mathit{rf}]
  {\Gamma', \pseq {\pseq \Delta}}
  {
   \infer[\mathit{rf}]
   {\fseq{\Gamma'}, \pseq \Delta}
   {\fseq{\fseq{\Gamma'}}, \Delta}
  }
 }
}
& 
\qquad
\deduce{\deduce{ \fseq{\fseq {\Gamma'}}, \Delta}{\vdots}}
{
 \infer[\mathit{rp}]
 {\fseq {\fseq{\Gamma'}}, \Delta'}
 {
   \infer[\mathit{ctr}]
   {\fseq{\Gamma'}, \pseq{\Delta'}}
   {
    \infer[\mathit{cut}]
    {\fseq{\Gamma'}, \fseq{\Gamma'}, \pseq{\Delta'}}
    {
     \deduce{\fseq{\Gamma'}, A_1}{\vdots}
     &
     \infer[\mathit{cut}]
     {\lneg{A_1}, \fseq{\Gamma'}, \pseq{\Delta'}}
     {
       \deduce{\fseq{\Gamma'}, A_2}{\vdots}
       &
       \deduce{\lneg{A_1}, \lneg{A_2}, \pseq{\Delta'}}{\vdots}
     }
    }
   }
 }
} \\ \\
(4) & \qquad (5)
\end{array}
$$
\caption{Some derivations in $\SKt$ illustrating the basic idea
of cut elimination}
\label{fig:cut}
\end{figure}

The main difficulty in proving cut elimination for $\SKt$ is in
finding the right cut reduction for some cases involving
the rules $\mathit{rp}$ and $\mathit{rf}$. For instance, consider the
derivation (1) in Figure~\ref{fig:cut}.
It is not obvious that there is a cut reduction strategy
that works locally without generalizing the
cut rule to, e.g., one which allows cut on any sub-sequent
in a sequent. Instead, we shall follow a global cut reduction
strategy similar to that used in cut elimination for display logics~\cite{Belnap82JPL}.
The idea is that, instead of permuting the cut rule locally,
we trace the cut formula $A$ (in $\Pi_1$) and $\lneg A$ (in $\Pi_2$),
until they both become principal in their respective proofs, 
and then apply the cut rule(s) at that point on smaller formulae. 
Schematically, our simple strategy can be illustrated as follows:
Suppose that $\Pi_1$  and $\Pi_2$ are, respectively, derivation (2) and (3) 
in Figure~\ref{fig:cut}, 
that $A = A_1 \land A_2$ and there is a single instance in each proof
where the cut formula is used. 
To reduce the cut on $A$, we first transform $\Pi_1$ by uniformly
substituting $\pseq \Delta$ for $A$ in $\Pi_1$ (see 
derivation (4) in Figure~\ref{fig:cut}).
We then prove the open leaf $\fseq{\fseq{\Gamma'}}, \Delta$
by uniformly substituting ${\fseq{\Gamma'}}$ for $\overline A$ in $\Pi_2$
(see derivation (5) in Figure~\ref{fig:cut}).
Notice that the cuts on $A_1$ and $A_2$ 
introduced in the proof above are on smaller
formulae than $A$.

The above simplified explanation implicitly assumes that a uniform
substitution of a formula (or formulae) in a derivation results in a
well-formed derivation, and that the cut formulae are not contracted.  
The precise statement of the proof substitution idea becomes more involved
once these aspects are taken into account. This will be made precise in
the main lemmas in the cut elimination proof.

Note that the proof substitution technique outlined above can actually
be applied to proof systems that are more general than $\SKt$; what
is essentially needed is that the inference rules of the proof systems
obey a certain closure property under arbitrary substitutions of
structures for formulae. 
In the following, in anticipation of extensions of $\SKt$ to be
presented in Section~\ref{sec:ext}, we shall prove a more general
cut elimination statement, which applies to any extensions of $\SKt$
with a certain class of structural rules. 

\begin{definition}
Let $\Gamma$ be a nested sequent. We denote with $\mathcal{F}(\Gamma)$
the multiset of all formula occurrences in $\Gamma$. 
A structural rule $\rho$ is said to be {\em linear} if for every
instance of the rule
$$
\infer[\rho]
{\Gamma}{\Delta}
$$
we have that $\mathcal{F}(\Gamma) = \mathcal{F}(\Delta).$ That is, a linear
structural rule does not allow weakening or contraction of
formulae occurrences in the premise or conclusion of the rule. 
We shall assume that each linear rule induces, for each of 
its instance, a bijection between formula occurrences in the premise
and formula occurrences in the conclusion, so that in every instance of the rule,
a formula occurrence in the premise can be related to a unique formula occurrence in the
conclusion, and vice versa.\footnote{To guarantee that such a bijection does exist for each instance, 
we shall restrict to only inference rules which can be represented as finite schemata
with no side conditions, as are commonly found in most proof systems.
}
A linear structural rule $\rho$ is said to be {\em substitution-closed}
if for any instance of the rule as given below left, 
where $A$ is a formula occurrence shared between the premise and the conclusion,
one can obtain another instance of $\rho$ as given below right, for any
structure $\Delta$:
$$
\infer[\rho]
{\Sigma[A]}
{\Sigma'[A]}
\qquad
\infer[\rho]
{\Sigma[\Delta]}
{\Sigma'[\Delta]}
$$
\end{definition}

The substitution-closure property mentioned above is 
similar to Belnap's condition (C6) for cut elimination for display logics \cite{Belnap82JPL}.
Note that this requirement for substitution closure  
rules out context-sensitive linear rules such as the rule shown in the leftmost figure below.
To see why, consider the instance of $\rho$ shown in the middle figure below. 
If one substitutes $\fseq a$ for one of the occurrences of $b$, say, the first one
from the left, then the resulting instance, as shown in the rightmost derivation below, would not
be a valid instance of $\rho.$
$$
\infer[\rho ]
{\Gamma, \fseq \Delta, \Delta}
{\Gamma, \pseq \Delta, \Delta}
\qquad
\infer[\rho ]
{a, \fseq {b, c}, {b, c}}
{a, \pseq {b, c}, {b, c}}
\qquad
\infer[]
{a, \fseq {\fseq a, c}, {b, c}}
{a, \pseq {\fseq a, c}, {b, c}}
$$

We use the notation $\vdash_S \Gamma$ to denote that the sequent $\Gamma$
is derivable in the proof system $S$. We write $\vdash_S \Pi : \Gamma$ 
when we want to be explicit about the particular derivation $\Pi$ of $\Gamma.$
The {\em cut rank} of an instance of cut 
is defined as usual as the size of the cut formula. 
The cut rank of a derivation $\Pi$, denoted with $cr(\Pi)$, 
is the largest cut rank of the cut instances
in $\Pi$ (or zero, if there are no cuts in $\Pi$). 
Given a formula $A$, we denote with $|A|$ its size. 
Given a derivation $\Pi$, we denote with $|\Pi|$ its height, i.e.,
the length of a longest branch in the derivation tree of $\Pi.$

We shall now give a general cut elimination proof for any extension
of $\SKt$ with substitution-closed linear structural rules. 
So in the following lemmas and theorem, we shall assume a (possibly empty)
set $\Scal$ of substitution-closed linear structural rules.
We denote with $\SKt+\Scal$ the proof system obtained by adding
the rules in $\Scal$ to $\SKt.$

\begin{lemma}
\label{lm:atom}
If $\vdash_{\SKt+\Scal} \Pi_1 : \Gamma, a$ and $\vdash_{\SKt+\Scal} \Pi_2 : \Sigma^k[\bar a]$,
where $k \geq 1$ and both $\Pi_1$ and $\Pi_2$ are cut free, then there exists 
a cut free $\Pi$ such that  $\vdash_{\SKt+\Scal} \Pi : \Sigma^k[\Gamma].$
\end{lemma}
\begin{proof}
By induction on $|\Pi_2|$. For the base cases, the non-trivial case is
when $\Pi_2$ ends with $id$ and $\bar a$ is active in the rule, i.e.,
$\Sigma^k[\bar a] = \Sigma_1^{k-1}[\bar a], \bar a, a$ and $\Pi_2$ is
as shown below left. Then we construct $\Pi$ as shown below right. 
$$
\infer[\mathit{id}]
{\Sigma_1^{k-1}[\bar a], \bar a, a}
{}
\qquad \qquad
\infer[\mathit{wk}]
{\Sigma_1^{k-1}[\Gamma], \Gamma, a}
{
 \deduce{\Gamma, a}{\Pi_1}
}
$$  
Most of the inductive cases follow straightforwardly from the induction hypothesis.
We show here two non-trivial cases involving contraction and a rule in $\Scal$: 
\begin{enumerate}[$\bullet$]
\item Suppose $\Sigma^k[\bar a] = \Sigma_1^i[\bar a], \Sigma_2^j[\bar a]$
and $\Pi_2$ ends with a contraction on $\Sigma_2^j[\bar a]$, as shown
below left. Then $\Pi$ is constructed as shown below right,
where $\Pi_2''$ is obtained from the induction hypothesis:
$$
\infer[\mathit{ctr}]
{\Sigma_1^i[\bar a], \Sigma_2^j[\bar a]}
{
 \deduce{\Sigma_1^i[\bar a], \Sigma_2^j[\bar a], \Sigma_2^j[\bar a]}{\Pi_2'}
}
\qquad \qquad
\infer[\mathit{ctr}]
{\Sigma_1^i[\Gamma], \Sigma_2^j[\Gamma]}
{
 \deduce{\Sigma_1^i[\Gamma], \Sigma_2^j[\Gamma], \Sigma_2^j[\Gamma]}{\Pi_2''}
}
$$

\item Suppose $\Pi_2$ is as shown below left, 
where $\rho \in \Scal.$ Then $\Pi$ is constructed as shown below right, 
where $\Pi_2''$ is obtained from the induction hypothesis:
$$
\infer[\rho]
{\Sigma^k[\bar a]}
{\deduce{\Sigma'^k[\bar a]}{\Pi_2'}}
\qquad \qquad
\infer[\rho .]
{\Sigma^k[\Gamma]}
{\deduce{\Sigma'^k[\Gamma]}{\Pi_2''}}
$$
The substitution closure property of $\rho$ guarantees that
the instance of $\rho$ on the right is valid.
\end{enumerate}
\end{proof}

Note that for the substitution of proofs in Lemma~\ref{lm:atom} (and 
other substitution lemmas to follow) to succeed,
one needs to allow contraction on arbitrary structures. 
Note also that as the rules from $\Scal$ are closed under substitution of
structures for formulae, they do not require any special treatment
in the following proofs of substitution lemmas, i.e., in
inductive cases involving these rules, the properties being proved
can be established by straightforward applications of the inductive
hypotheses, so we shall not detail the cases involving these rules.

\begin{lemma}
\label{lm:or}
Suppose $\vdash_{\SKt+\Scal} \Pi_1 : \Delta, A$ and $\vdash_{\SKt+\Scal} \Pi_2: \Delta, B$ 
and $\vdash_{\SKt+\Scal} \Pi : \Sigma^k[\lneg A \lor \lneg B]$, for some $k \geq 1$, 
where the cut ranks of $\Pi_1$, $\Pi_2$ and $\Pi$ are 
smaller than $|A \land B|$.
Then there exists a proof $\Pi'$ such that
$ \vdash_{\SKt+\Scal} \Pi' : \Sigma^k[\Delta]$ and $cr(\Pi) < |A \land B|.$
\end{lemma}
\begin{proof}
By induction on $|\Pi|$. Most cases are straightforward. The only non-trivial
case is when $\lneg A \lor \lneg B$ is principal in the last rule of $\Pi$,
i.e., $\Pi$ is of the form shown below left. The proof $\Pi'$ is constructed
as shown below right, where $\Psi'$ is a cut-free derivation obtained 
via the induction hypothesis. 
$$
\infer[]
{\Sigma_1^{k-1}[\lneg A \lor \lneg B], \lneg A \lor \lneg B}
{
 \deduce{\Sigma_1^{k-1}[\lneg A \lor \lneg B], \lneg A, \lneg B}{\Psi}
}
\qquad
\infer[\mathit{ctr}]
{\Sigma_1^{k-1}[\Delta], \Delta}
{
 \infer[\mathit{cut}]
 {\Sigma_1^{k-1}[\Delta], \Delta, \Delta}
 {
   \deduce{\Delta, A}{\Pi_1}
   &
   \infer[\mathit{cut}]
   {\Sigma_1^{k-1}[\Delta], \lneg A, \Delta}
   {
     \deduce{\Delta, B}{\Pi_2}
     &
     \deduce{\Sigma_1^{k-1}[\Delta], \lneg A, \lneg B}{\Psi'}
   }
 }
}
$$
\end{proof}

\begin{lemma}
\label{lm:and}
Suppose $\vdash_{\SKt+\Scal} \Pi_1 : \Delta, A, B$ and 
$\vdash_{\SKt+\Scal} \Pi_2 : \Sigma^k[\lneg A \land \lneg B]$,
for some $k \geq 1$, 
and the cut ranks of $\Pi_1$ and $\Pi_2$ are smaller than $|A \lor B|$. 
Then there exists a proof $\Pi$ such that $\vdash_{\SKt+\Scal} \Pi : \Sigma^k[\Delta]$
and $cr(\Pi) < |A \lor B|.$  
\end{lemma}
\begin{proof}
This is proved analogously to Lemma~\ref{lm:or}. 
\end{proof}

To prove the next two lemmas, we use two derived rules, i.e.,
$d1$ and $d2$ given below. 
These two rules are derivable using $\mathit{rp}$, $\mathit{rf}$, $\mathit{ctr}$ and $\mathit{wk}.$
They are similar to the so-called ``medial rules'' used to
prove admissibility of structure contraction in \cite{Brunnler09tableaux}.
The rule $d1$ is derived as shown in the rightmost derivation below
($d2$ is derived analogously).
$$
\infer[\mathit{d1}]
{\Gamma, \fseq {\Delta_1, \Delta_2}}
{\Gamma, \fseq {\Delta_1}, \fseq {\Delta_2}}
\qquad
\infer[\mathit{d2}]
{\Gamma, \pseq {\Delta_1, \Delta_2}}
{\Gamma, \pseq {\Delta_1}, \pseq {\Delta_2}}
\qquad
\infer[\mathit{ctr}]
{\Gamma, \fseq{\Delta_1, \Delta_2}}
{
\infer[\mathit{rp}]
{\Gamma, \fseq{\Delta_1, \Delta_2}, \fseq{\Delta_1, \Delta_2}}
{
\infer[\mathit{wk}]
{\pseq{\Gamma, \fseq{\Delta_1, \Delta_2}}, \Delta_1, \Delta_2}
{
\infer[\mathit{rf}]
{\pseq{\Gamma, \fseq{\Delta_1, \Delta_2}}, \Delta_1}
{
\infer[\mathit{rp}]
{\Gamma, \fseq{\Delta_1, \Delta_2}, \fseq{\Delta_1}}
{
\infer[\mathit{wk}]
{\Delta_1, \Delta_2, \pseq{\Gamma, \fseq{\Delta_1}}}
{
\infer[\mathit{rf}]
{\Delta_2, \pseq{\Gamma, \fseq{\Delta_1}}}
{\Gamma, \fseq{\Delta_1}, \fseq{\Delta_2}}
}
}
}
}
}}
$$

\begin{lemma}
\label{lm:fdia}
Suppose $\vdash_{\SKt+\Scal} \Pi_1 : \Delta, \fseq A$ and 
$\vdash_{\SKt+\Scal} \Pi_2 : \Sigma^k[\FDia \lneg A]$, for some $k \geq 1$,
and the cut ranks of $\Pi_1$ and $\Pi_2$ are smaller than $|\FBox A|$. 
Then there exists a proof $\Pi$ such that
$\vdash_{\SKt+\Scal} \Pi : \Sigma^k[\Delta]$ and $cr(\Pi) < |\FBox A|.$
\end{lemma}
\begin{proof}
By induction on $|\Pi_2|$. The non-trivial case is when $\Pi_2$ ends with $\FDia$
on $\FDia \lneg A$, as shown below left. 
The derivation $\Pi$ in this case is constructed as shown below right. 
There, the derivation $\Pi'$ is obtained by applying the induction hypothesis to $\Pi_2'$.
Note that by the induction hypothesis, $cr(\Pi') < |\FBox A|.$ 
$$
\infer[\FDia]
{\Sigma_1^{k-1}[\FDia \lneg A], \fseq \Gamma, \FDia \lneg A}
{
\deduce{\Sigma_1^{k-1}[\FDia \lneg A], \fseq{\Gamma, \lneg A}}{\Pi_2'}
}
\qquad
\infer[\mathit{rp}]
{\Sigma_1^{k-1}[\Delta], \fseq \Gamma, \Delta}
{
 \infer[\mathit{d2}]
 {\pseq{\Sigma_1^{k-1}[\Delta], \Delta}, \Gamma}
 {
   \infer[\mathit{cut}]
   {\pseq{\Sigma_1^{k-1}[\Delta]}, \pseq{\Delta}, \Gamma}
   {
      \infer[\mathit{rf}]
      {\pseq{\Sigma_1^{k-1}[\Delta]}, \Gamma, \lneg A}
      {\deduce{\Sigma_1^{k-1}[\Delta], \fseq{\Gamma, \lneg A}}{\Pi'}}
      &
      \infer[\mathit{rf}]
      {\pseq{\Delta}, A}
      {\deduce{\Delta, \fseq A}{\Pi_1}}
   }
 }
}
$$
\end{proof}

\begin{lemma}
\label{lm:fbox}
Suppose $\vdash_{\SKt+\Scal} \Pi_1 : \Delta, \fseq {\Delta', A}$ and
$\vdash_{\SKt+\Scal} \Pi_2 : \Sigma^k[\FBox \lneg A]$, for some $k \geq 1$,
and the cut ranks of $\Pi_1$ and $\Pi_2$ are smaller than $|\FDia A|.$
Then there exists $\Pi$ such that 
$\vdash_{\SKt+\Scal} \Pi : \Sigma^k[\Delta, \fseq{\Delta'}]$
and $cr(\Pi) < |\FDia A|.$
\end{lemma}
\begin{proof}
By induction on $|\Pi_2|.$
The non-trivial case $\Pi_2$ is when $\Pi_2$ is as given below left. 
The derivation $\Pi$ is constructed as shown below right, where $\Pi'$
is obtained from the induction hypothesis and which satisfies
$cr(\Pi') < |\FDia A|.$
$$
\infer[\FBox]
{\Sigma_1^{k-1}[\FBox \lneg A], \FBox {\lneg A}}
{
\deduce{\Sigma_1^{k-1}[\FBox \lneg A], \fseq {\lneg A}}{\Pi_2'}
}
\qquad
\infer[\mathit{rp}]
{\Sigma_1^{k-1}[\Delta, \fseq {\Delta'}], \Delta, \fseq {\Delta'}}
{
  \infer[\mathit{d2}]
  {\pseq {\Sigma_1^{k-1}[\Delta, \fseq {\Delta'}], \Delta}, \Delta'}
  {
    \infer[\mathit{cut}]
    {\pseq {\Sigma_1^{k-1}[\Delta, \fseq {\Delta'}]}, \pseq \Delta, \Delta'}
    {
      \infer[\mathit{rf}]
      {\pseq {\Sigma_1^{k-1}[\Delta, \fseq {\Delta'}]}, \lneg A}
      {\deduce{\Sigma_1^{k-1}[\Delta, \fseq {\Delta'}], \fseq {\lneg A}}{\Pi'}}
      &
      \infer[\mathit{rf}]
      {\pseq \Delta, \Delta', A}
      {\deduce{\Delta, \fseq{\Delta', A}}{\Pi_1}}
    }
  }
}
$$
\end{proof}

\begin{lemma}
\label{lm:pdia}
Suppose $\vdash_{\SKt+\Scal} \Pi_1 : \Delta, \pseq A$ and 
$\vdash_{\SKt+\Scal} \Pi_2 : \Sigma^k[\PDia \lneg A]$, for some $k \geq 1$,
and the cut ranks of $\Pi_1$ and $\Pi_2$ are smaller than $|\PBox A|$. 
Then there exists a proof $\Pi$ such that
$\vdash_{\SKt+\Scal} \Pi : \Sigma^k[\Delta]$ and $cr(\Pi) < |\PBox A|.$
\end{lemma}
\begin{proof}
This is proved analogously to Lemma~\ref{lm:fdia}.
\end{proof}

\begin{lemma}
\label{lm:pbox}
Suppose $\vdash_{\SKt+\Scal} \Pi_1 : \Delta, \pseq {\Delta', A}$ and
$\vdash_{\SKt+\Scal} \Pi_2 : \Sigma^k[\PBox \lneg A]$, for some $k \geq 1$,
and the cut ranks of $\Pi_1$ and $\Pi_2$ are smaller than $|\PDia A|.$
Then there exists $\Pi$ such that $\vdash_{\SKt+\Scal} \Pi : \Sigma^k[\Delta, \pseq{\Delta'}]$
and $cr(\Pi) < |\PDia A|.$
\end{lemma}
\begin{proof}
This is proved analogously to Lemma~\ref{lm:fbox}.
\end{proof}

\begin{lemma}
\label{lm:cutelim-step}
Let $C$ be a non-atomic formula. 
Suppose $ \vdash_{\SKt+\Scal} \Psi_1 : \Gamma, \lneg C$ and 
$\vdash_{\SKt+\Scal} \Psi_2 : \Omega^n[C]$,
for some $n \geq 1$, and the cut ranks of $\Psi_1$ and $\Psi_2$ are smaller than $|C|.$
Then there exists a proof $\Psi$ such that 
$\vdash_{\SKt+\Scal} \Psi : \Omega^n[\Gamma]$ and $cr(\Psi) < |C|.$
\end{lemma}
\begin{proof}
By induction on the height of $\Psi_2$ and case analysis on $C$. The
non-trivial cases are when $\Psi_2$ ends with an introduction rule on
$C$. That is, we have $\Omega^n[C] = \Omega_1^{n-1}[C], C$ for some
context $\Omega_1^{n-1}[]$. We show the cases where $C$ is either
$\FBox B$, $\FDia B$ or $B_1 \land B_2$; the other cases can
be treated similarly.
\begin{enumerate}[$\bullet$]
\item Suppose $C = \FBox B$ and $\Psi_2$ is the following derivation:

$$
\infer[\FBox]
{\Omega_1^{n-1}[\FBox B], \FBox B}
{
\deduce{\Omega_1^{n-1}[\FBox B], \fseq B}{\Psi_2'}
}
$$
By induction hypothesis, we have 
$\vdash_{\SKt+\Scal} \Psi' : \Omega_1^{n-1}[\Gamma], \fseq B$ and $cr(\Psi') < |C|$. 
Applying Lemma~\ref{lm:fdia} to $\Psi'$ and $\Psi_1$ 
(that is, by instantiating $A$ to $B$, $\Delta$ to $\Omega_1^{n-1}[\Gamma],$
and $\Sigma^k[]$ to the context $\Gamma, [~]$),
we obtain $\vdash_{\SKt+\Scal} \Psi : \Gamma, \Omega_1^{n-1}[\Gamma] =
\Omega^{n}[\Gamma]$ such that $cr(\Psi) < | \FBox B|$.

\item Suppose $C = \FDia B$ and $\Psi_2$ is the following derivation:
$$
\infer[\FDia]
{\Omega_1^{n-1}[\FDia B], \fseq{\Gamma'}, \FDia B}
{
\deduce{\Omega_1^{n-1}[\FDia B], \fseq{\Gamma', B}}{\Psi_2'}
}
$$
By induction hypothesis, we have 
$$\vdash_{\SKt+\Scal} \Psi' : \Omega_1^{n-1}[\Gamma], \fseq{\Gamma', B}.$$
Applying Lemma~\ref{lm:fbox} to $\Psi'$ and $\Psi_1$
(i.e., instantiating $A$ to $B$,
$\Delta$ to $\Omega^{n-1}[\Gamma]$, $\Delta'$ to $\Gamma'$, 
and $\Sigma^k[]$ to the context $\Gamma, []$), we obtain
$\vdash_{\SKt+\Scal} \Psi : \Gamma, \Omega_1^{n-1}[\Gamma] =
 \Omega^{n}[\Gamma]$ such that $cr(\Psi) < | \FDia B|$.

\item Suppose $C = B_1 \land B_2$ and $\Psi_2$ is the following derivation:
$$
\infer[\land]
{\Omega_1^{n-1}[B_1 \land B_2], B_1 \land B_2}
{
\deduce{\Omega_1^{n-1}[B_1 \land B_2], B_1}{\Theta_1}
&
\deduce{\Omega_1^{n-1}[B_1 \land B_2], B_2}{\Theta_2}
}
$$
By induction hypothesis, we have 
$\vdash_{\SKt+\Scal} \Theta_1' : \Omega_1^{n-1}[\Gamma], B_1$ and 
$\vdash_{\SKt+\Scal} \Theta_2' : \Omega_1^{n-1}[\Gamma], B_2$.  
Applying Lemma~\ref{lm:or} to $\Theta_1'$, $\Theta_2'$ and 
$\Psi_1$, we obtain 
$\vdash_{\SKt+\Scal} \Psi : \Gamma, \Omega_1^{n-1}[\Gamma] = \Omega^{n}[\Gamma]$
such that $cr(\Psi) < |B_1 \land B_2|$.
\end{enumerate}
\end{proof}

\begin{theorem}
\label{thm:cut-elim}
Cut elimination holds for $\SKt + \Scal.$
\end{theorem}
\begin{proof}
Given a derivation with cuts, we remove topmost 
cuts in succession, using Lemma~\ref{lm:atom} and Lemma~\ref{lm:cutelim-step}. 
\end{proof}

\begin{corollary}
Cut elimination holds for $\SKt.$
\end{corollary}

\section{System \texorpdfstring{$\DKt$}{DKt}: a contraction-free deep-sequent calculus}
\label{sec:dkt}

We now consider another sequent system which uses {\em deep inference}, where
rules can be applied directly to any node within a nested sequent.
We call this system $\DKt$, and give its inference rules in Figure~\ref{fig:DKt}.
Note that there are no structural rules in $\DKt$, and the contraction rule is
absorbed into the logical rules. Notice that, reading the logical rules bottom up,
we keep the principal formulae in the premise. This is actually not neccessary
for some rules (e.g., $\PBox$, $\land$, etc.), but this form of rule allows
for a better accounting of formulae in our saturation-based proof search procedure
(see Section~\ref{sec:search}).
We also do not include the cut rule in $\DKt$ as it is admissible in $\DKt$,
the translations from $\SKt$ to $\DKt$ and back, to be shown below, do not
use the cut rule. 
A side note on the cut rule: one could introduce 
a ``deep'' version of cut:
$$
\infer[\mathit{cut},]
{\Sigma[\emptyset]}
{\Sigma[A] & \Sigma[\bar A]}
$$
just as is done in nested sequent calculi for modal logics 
in~\cite{brunnler2006,Brunnler09tableaux}. This form of cut rule can be
easily derived from its shallow counterpart (see Figure~\ref{fig:SKt}) 
using the display property (Proposition~\ref{prop:display}). So when we speak of
cut admissibility in $\DKt$, it applies equally to both the shallow cut
and the deep cut above. 

The following intuitive observation about $\DKt$ rules will be useful later: 
Rules in $\DKt$ are characterized by propagations of formulae across different
nodes in a nested sequent tree. The shape of the tree is not affected by
these propagations, and the only change that can occur to the tree is the creation
of new nodes (via the introduction rules $\PBox$ and $\FBox$).

System $\DKt$ corresponds to Kashima's
$\S2Kt$~\cite{kashima-cut-free-tense}, but with the contraction rule
absorbed into the logical rules. 
The modal fragment of $\DKt$ was also developed independently
by Br\"unnler~\cite{brunnler2006,Brunnler09tableaux} and Poggiolesi~\cite{poggiolesi2009}.
Kashima shows that $\DKt$ proofs can
be encoded into $\SKt$, essentially due to the display
property of $\SKt$ (Proposition~\ref{prop:display}) which allows
displaying and undisplaying of any node within a nested sequent.
Kashima also shows that $\DKt$ is complete for tense logic, via
semantic arguments.  We prove a stronger result: every cut-free
$\SKt$-proof can be transformed into a $\DKt$-proof, hence $\DKt$ is
complete and cut is admissible in $\DKt$.

\begin{figure}[t]
$$
\begin{array}{c@{\qquad \quad}c@{\qquad \quad}c}
\infer[\mathit{id}]
{\Sigma[a,\bar a]}{}
&
\infer[\land]
{\Sigma[A \land B]}
{\Sigma[A \land B, A] & \Sigma[A \land B, B]}
&
\infer[\lor]
{\Sigma[A \lor B]}
{\Sigma[A \lor B, A, B]}
\\ \\ 
\infer[\PBox]
{\Sigma[\PBox A]}{\Sigma[\PBox A, \pseq A]}
&
\infer[\PDia_1]
{\Sigma[\pseq{\Delta}, \PDia A]}
{\Sigma[\pseq{\Delta, A}, \PDia A]}
&
\infer[\PDia_2]
{\Sigma[\fseq{\Delta, \PDia A}]}
{\Sigma[\fseq{\Delta, \PDia A}, A]}
\\ \\
\infer[\FBox]
{\Sigma[\FBox A]}{\Sigma[\FBox A, \fseq A]}
&
\infer[\FDia_1]
{\Sigma[\fseq{\Delta}, \FDia A]}
{\Sigma[\fseq{\Delta, A}, \FDia A]}
&
\infer[\FDia_2]
{\Sigma[\pseq{\Delta, \FDia A}]}
{\Sigma[\pseq{\Delta, \FDia A}, A]}
\end{array}
$$
\caption{The contraction-free deep-inference system $\DKt$}
\label{fig:DKt}
\end{figure}

To translate cut-free $\SKt$-proofs into $\DKt$-proofs, we 
show that all structural rules of $\SKt$ are height-preserving
admissible in $\DKt$. 
\begin{definition}
Given a proof system $\mbf S$ and a rule $\rho$ with 
premises $\Gamma_1,\ldots,\Gamma_n$ and 
conclusion $\Gamma$, $\rho$ is said to be 
{\em admissible in $\mbf S$} if the following holds:
whenever $\vdash_{\mbf S} \Pi_1 : \Gamma_1$, 
$\ldots, \vdash_{\mbf S} \Pi_n : \Gamma_n$, then
there exists $\Pi$ such that $\vdash_{\mbf S} \Pi : \Gamma.$
In the case where $n=1$, we say that $\rho$ 
is {\em height-preserving admissible in $\mbf S$}
if $|\Pi| = |\Pi_1|.$
\end{definition}
In the following lemmas, we show a stronger admissibility result for
weakening and contraction, i.e., we shall show that the following {\em deep} versions
of weakening and contraction are in fact admissible. 
$$
\infer[\mathit{dw}]
{\Sigma[\Gamma,\Delta]}
{\Sigma[\Gamma]}
\qquad
\infer[\mathit{dgc}]
{\Sigma[\Delta]}
{\Sigma[\Delta,\Delta]}
$$
Obviously, the rules $wk$ and $ctr$ are just instances of the above
rules. 
As we shall see, admissibility of $\mathit{dgc}$ follows from admissibility of
formula contraction (the rule $\mathit{dfc}$ below) and two distribution rules shown below.
$$
\infer[\mathit{dfc}]
{\Sigma[A]}
{\Sigma[A,A]}
\qquad
\infer[\mathit{mf}]
{\Sigma[\fseq{\Delta_1,\Delta_2}]}
{\Sigma[\fseq{\Delta_1}, \fseq{\Delta_2}]}
\qquad
\infer[\mathit{mp}]
{\Sigma[\pseq{\Delta_1,\Delta_2}]}
{\Sigma[\pseq{\Delta_1}, \pseq{\Delta_2}]}
$$
The distribution rules $\mathit{mf}$ and $\mathit{mp}$ 
are usually called the {\em medial} rules in 
the deep inference literature 
(see, e.g., 
\cite{Brunnler01LPAR,Gore07JLC,Brunnler09tableaux}),
and, in their various forms, they have been used to reduce general contraction to
formulae or atomic contraction in different proof systems
for classical, intuitionistic, linear, modal and tense logics. 
The modal medial rule $\mathit{mf}$ has been used in \cite{Brunnler09tableaux} to
show admissibility of contraction for several nested sequent calculi 
for modal logics. Our proof of admissibility of contraction here is 
an extension of Br\"unnler and Stra{\ss}burger's 
proof~\cite{Brunnler09tableaux} to tense logics.

\begin{lemma}[Admissibility of weakening]
\label{lm:weak}
The rule $dw$ is height-preserving admissible in $\DKt.$ 
\end{lemma}
\begin{proof}
By simple induction on $|\Pi|.$
\end{proof}

The proofs for the following lemmas that concern structural rules that
change the shape of the tree of a nested sequent share
similarities. That is, the only interesting cases in the proofs are
those that concern propagation of formulae across different nodes in a
nested sequent. We show here an interesting case in the proof for the
admissibility of display postulates.

\begin{lemma}[Admissibility of display postulates]
\label{lm:rp-rf}
The rules $\mathit{rp}$ and $\mathit{rf}$ are both height-preserving admissible in
$\DKt.$
\end{lemma}
\begin{proof}
We show here admissibility of $rp$, the other rule can be dealt with similarly. 
Consider the $rp$ rule in Figure~\ref{fig:SKt}. 
Suppose that $\vdash_\DKt \Pi : \Gamma, \pseq{\Delta}.$
We shall construct a derivation $\Pi'$ for the nested sequent 
$\fseq \Gamma, \Delta$ by induction on $|\Pi|.$
The non-trivial cases are when there is an exchange of formulae 
between $\Gamma$ and $\Delta$. We show one case below; the others
can be done analogously. 
Suppose $\Pi$ is as shown below left, where $\Gamma = \Gamma', \PDia A$. 
Then $\Pi'$ is as shown below right
where $\Pi_1'$ is obtained from the induction hypothesis:
$$
\infer[\PDia_1]
{\Gamma', \PDia A, \pseq{\Delta}}
{
\deduce{\Gamma', \PDia A, \pseq{A, \Delta}}{\Pi_1}
}
\qquad\qquad
\infer[\PDia_2]
{\fseq{\Gamma', \PDia A}, \Delta}
{
\deduce{\fseq{\Gamma', \PDia A}, A, \Delta}{\Pi_1'}
}
$$
By the induction hypothesis $|\Pi_1| = |\Pi_1'|$,
hence we also have $|\Pi| = |\Pi'|.$
\end{proof}

To show admissibility of general contraction, we first 
show that formula contraction, $\mathit{mp}$ and $\mathit{mf}$
are all heigh-preserving admissible in $\DKt.$

\begin{lemma}
\label{lm:medial}
The rules $\mathit{dfc,mf}$ and $\mathit{mp}$ are height-preserving admissible
in $\DKt.$
\end{lemma}
\begin{proof}
Height-preserving admissibility of $\mathit{dfc}$ can be proved
by simple induction on the height of the derivation of its premise. 
We show here height-preserving admissibility of $\mathit{mf}$;
height-preserving admissibility of $\mathit{mp}$ can be proved
analogously. 

So suppose we have $\vdash_\DKt \Pi : \Sigma[\fseq {\Delta_1}, \fseq{\Delta_2}].$
We show by induction on $|\Pi|$ that there exists $\Pi'$ such that
$\vdash_\DKt \Pi' : \Sigma[\fseq{\Delta_1,\Delta_2}]$ and $|\Pi| = |\Pi'|.$
We show here two non-trivial cases:
\begin{enumerate}[$\bullet$]
\item Suppose $\Pi$ ends with $\FDia_1$ that moves a formula
into $\fseq{\Delta_1}$ when read upwards. That is, 
$$\Sigma[\fseq{\Delta_1}, \fseq{\Delta_2}] = 
\Sigma'[\FDia A, \fseq{\Delta_1}, \fseq{\Delta_2}]$$
and $\Pi$ is as shown below left. 
Then $\Pi'$ is constructed as shown below right, where 
$\Psi'$ is obtained by applying the induction hypothesis to $\Psi.$
$$
\infer[\FDia_1]
{\Sigma'[\FDia A, \fseq{\Delta_1},\fseq{\Delta_2}]}
{
 \deduce{\Sigma'[\FDia A, \fseq{A, \Delta_1}, \fseq{\Delta_2}]}{\Psi}
}
\qquad \qquad
\infer[\FDia_1]
{\Sigma'[\FDia A, \fseq{\Delta_1, \Delta_2}]}
{
\deduce{\Sigma'[\FDia A, \fseq{A, \Delta_1, \Delta_2}]}{\Psi'}
}
$$
Since $|\Psi'| = |\Psi|$, it follows that
$|\Pi'| = |\Pi|.$

\item Suppose $\Pi$ ends with $\PDia_2$ that moves a formula
out from $\fseq{\Delta_1}$. That is, $\Delta_1 = \PDia A, \Delta_1'$
and $\Pi$ is as shown below left. Then $\Pi'$ is constructed as 
shown below right, where $\Psi'$ is obtained from the induction hypothesis. 
It is easy to see that $|\Pi'| = |\Pi|.$
$$
\infer[\PDia_2]
{\Sigma[\fseq{\PDia A, \Delta_1'}, \fseq{\Delta_2}]}
{
 \deduce{\Sigma[A, \fseq{\PDia A, \Delta_1'}, \fseq{\Delta_2}]}{\Psi}
}
\qquad \qquad
\infer[\PDia_2]
{\Sigma[\fseq{\PDia A, \Delta_1', \Delta_2}]}
{
 \deduce{\Sigma[A, \fseq{\PDia A, \Delta_1, \Delta_2}]}{\Psi'}
}
$$
\end{enumerate}
\end{proof}

\begin{lemma}[Admissibility of contraction]
\label{lm:contr}
The rule $\mathit{dgc}$ is height-preserving admissible in $\DKt.$
\end{lemma}
\begin{proof}
Suppose $\vdash_\DKt \Pi : \Sigma[\Delta, \Delta].$
We need to show that there exists $\Pi'$ such that 
$\vdash_\DKt \Pi' : \Sigma[\Delta]$ and $|\Pi| = |\Pi'|.$
We do this by induction on the size of $\Delta$. If $\Delta$
is the empty set then it is straightforward. 
If $\Delta$ is a formula, then it is an instance of $\mathit{dfc}$
which is height-preserving admissible by Lemma~\ref{lm:medial}. 
The other cases follow from the induction hypothesis and
Lemma~\ref{lm:medial}.
Consider, for instance, the case where $\Delta = \fseq{\Delta'}$.
Then by Lemma~\ref{lm:medial} we have a proof $\Psi$, with
$|\Psi| = |\Pi|$, 
such that
$
\vdash_\DKt \Psi : \Sigma[\fseq{\Delta',\Delta'}].
$
Note that since $\Delta'$ is of a smaller size than $\fseq{\Delta'}$,
we can apply the induction hypothesis to $\Psi$ and obtain
a proof $\Pi'$, with $|\Pi'| \leq |\Pi|$, such that 
$
\vdash_\DKt \Pi' : \Sigma[\fseq{\Delta'}].
$
\end{proof}

\begin{theorem}
\label{thm:SKt-equal-DKt}
For every sequent $\Gamma$,
$\vdash_\SKt \Gamma$ if and only if $\vdash_\DKt \Gamma.$
\end{theorem}
\begin{proof}
The forward direction, that is, showing that $\vdash_\SKt \Gamma$
implies $\vdash_\DKt \Gamma$, 
follows from admissibility of the structural rules of $\SKt$ in
$\DKt$ (Lemma~\ref{lm:weak} -- Lemma~\ref{lm:contr}). 

For the converse, we use the display property of $\SKt$
(Proposition~\ref{prop:display}) to simulate the deep-inference rules of $\DKt$.
We show here the derivations for the rules $\FDia_1$ and $\PDia_2$ (the other
cases are similar):
$$
\deduce{\Sigma[\fseq{\Delta}, \FDia A]}
{
\deduce{\vdots}
{
 \infer[ctr]
 {\Delta', \fseq{\Delta}, \FDia A}
 {
  \infer[\FDia]
  {\Delta', \fseq{\Delta}, \FDia A, \FDia A}
  {
   \deduce{\Delta', \fseq{\Delta, A}, \FDia A}
   {
    \deduce{\vdots}{\Sigma[\fseq{\Delta, A}, \FDia A]}
   }
  }
 }
}
}
\qquad \qquad
\deduce{\Sigma[\fseq{\Delta, \PDia A}]}
{
\deduce{\vdots}
{
 \infer[rp]
 {\Delta', \fseq{\Delta, \PDia A}}
 {
  \infer[ctr]
  {\pseq{\Delta'}, \Delta, \PDia A}
  {
   \infer[\PDia]
   {\pseq{\Delta'}, \PDia A, \Delta,  \PDia A}
   {
    \infer[rf]
    {\pseq{\Delta', A}, \Delta, \PDia A }
    {
     \deduce
     {\Delta', A, \fseq{\Delta, \PDia A}}
     {\deduce{\vdots}{\Sigma[\fseq{\Delta, \PDia A}, A]}}
    }
   }
  }
 }
}
}
$$
where the dotted part of the derivation is obtained from applying
Proposition~\ref{prop:display}.
\end{proof}

A consequence of Theorem~\ref{thm:SKt-equal-DKt} is that the
general contraction rule in $\SKt$ can be replaced by formula contraction.
This can be proved as follows: take a cut-free proof in $\SKt$, translate
it to $\DKt$ and then translate it back to $\SKt$. Since general contraction
is admissible in $\DKt$, and since the translation from $\DKt$ to $\SKt$ does not
use general contraction (only formula contraction), we can effectively
replace the general contraction in $\SKt$ with formula contraction. 

An interesting feature of $\DKt$ is that 
in a proof of a sequent, the `colour' of a (formula or structural) connective 
does not change when moving from premise to conclusion or vice versa.
Let us call a formula (a sequent, a rule) {\em purely modal} if 
it contains no black connectives. It is easy to see that if
a purely modal formula (sequent) is provable in $\DKt$, then it is provable
using only purely modal rules. 
Let $\DK = \{id, \land, \lor, \FBox, \FDia_1\},$ i.e., it is the set of
purely modal rules of $\DKt.$ The above observation leads to the 
following ``separation'' result:
\begin{theorem}
\label{thm:separation}
For every modal formula $A$,  
$\vdash_{\DK} A$ iff $A$  is a theorem of $\myK$.
\end{theorem}
\begin{proof}
  ($\Rightarrow$) Suppose $\vdash_{\DK} A$. 
  Since $\DK$ is a subsystem of $\DKt$, we must have $\vdash_{\DKt} A$, 
  and then $\vdash_{\SKt} A$. By the soundness of
  $\SKt$, $A$ is Kt-valid. But all purely modal $Kt$-valid
  formulae are also $K$-valid. Thus purely modal $A$ 
  is also a theorem of $\myK$.

  ($\Leftarrow$) 
  Suppose $A$ is a theorem of
  $\myK$.  But the theorems of $\myK$ are
  also theorems of $\Kt$, hence $A$ is derivable in $\SKt$. 
  This derivation may contain cuts, but
  by cut elimination we know that $A$ is also cut-free derivable
  in $\SKt$. The cut-free $\SKt$-derivation of a purely modal formula
  cannot contain any instances of the rules $\PBox$ or $\PDia$ since
  these introduce non-modal connectives into their conclusion. Thus,
  the only way to create an occurrence of $\bullet$ on our way up from
  the end-sequent is to use $rp$. By Theorem~\ref{thm:SKt-equal-DKt},
  the cut-free $\SKt$-derivation of $A$ can be transformed into
  a (cut-free) derivation of $A$ in $\DKt$. Moreover, the
  transformation given in the proof removes all applications of $rp$
  without creating black structural or logical connectives. For
  example, an $\SKt$ derivation of $a, \lneg{a}, \pseq{\Delta}$ is
  converted to a $\DKt$ derivation of 
  $(\fseq{a, \lneg{a}}, \Delta) = \Sigma[a, \lneg{a}]$. Hence the
  transformed derivation is actually a derivation in $\DK$.
\end{proof}
This completeness result for $\DK$ is known from \cite{brunnler2006};
what we show here is how it can be derived as a consequence of 
completeness of $\DKt.$

\section{Proof systems for some extensions of tense logic}
\label{sec:ext}

We now consider extensions of tense logic with 
a class of axioms that subsumes a range of standard normal modal
axioms, e.g., reflexivity, transitivity, euclideanness, etc. 
These axioms, called Scott-Lemmon axioms \cite{Lemmon77},
are formulae of the form:
$$
G(h,i,j,k): \qquad \FDia^h \FBox^i A \impl \FBox^j \FDia^k A
$$
where $h,i,j,k \geq 0$ and $\FDia^n A$ (likewise, $\FBox^n A$)
denotes the formula $A$ prefixed with $n$-occurrences of $\FDia$ (resp. $\FBox$).
For example, the axiom for transitivity,
$\FBox A \impl \FBox \FBox A$, is an instance of Scott-Lemmon
axiom scheme with $h = 0,$ $i=1$, $j=2$ and $k=0.$

In the following subsection, we show that, for each set $\mbf{SL}$
of Scott-Lemmon axioms, there is a shallow system that modularly 
extends $\SKt$ with $\mbf{SL}$ for which
cut elimination holds. By modular extension we mean that
the rules of the extended systems are the rules of $\SKt$
plus a set of structural rules that are derived directly
from the modal axioms (in fact, they are in one-to-one correspondence).
However, there does not appear to be a systematic way to derive
the corresponding deep-inference systems for these extensions. 
In subsequent subsections, we give deep-inference systems for 
two well-known extensions of $\Kt$, 
i.e., $\Kt$ extended with axioms for $S4$ and $S5$, and
an extension of $\Kt$ with the axiom of uniqueness $CD : \FDia A \impl \FBox A.$ 
Again, as with $\DKt$, the rules for the deep-inference systems are characterized
by propagations of formulae across different nodes in the nested sequents.
However, the design of the rules for the deep system is not as modular
as its shallow counterpart, since it needs to take into account
the closure of the axioms. 

A nice feature of the deep inference systems shown below is that 
they satisfy the same separation property as with $\DKt$: 
the purely modal subset of each deep-inference system is sound and complete
with respect to its modal fragment. That is, we obtain the deep-inference systems for
S4, S5 and $K + CD$ ``for free'' simply by dropping all the tense rules. 

\subsection{Extending \texorpdfstring{$\SKt$}{SKt} with Scott-Lemmon axioms}

One way to extend $\SKt$ with Scott-Lemmon axioms is to
simply add those axiom schemes as inference rules without premise.
However, the resulting system would not satisfy cut elimination.
Instead, we shall follow an approach that absorbs those axioms into
structural rules without breaking cut elimination. 
In the display calculus setting, Kracht~\cite{kracht-power} has shown that 
a class of axioms, called {\em primitive axioms}, can be turned
into structural rules in a systematic way and the display calculus
for tense logic extended with those structural rules also satisfies
cut elimination. 
A {\em primitive axiom} is an axiom of the form $A \impl B$ where both
$A$ and $B$ are built using propositional variables, $\land$, $\lor$,
$\FDia$, and $\PDia$. We shall follow Kracht's approach in
absorbing Scott-Lemmon's axioms into structural rules. However,
the main problem is that Scott-Lemmon axioms, in the form shown earlier,
are not strictly speaking primitive axioms. But as we shall see later,
they have equivalent representations in primitive form.  
A {\em primitive Scott-Lemmon axiom} is a formula of the form
$$
P(h,i,j,k): \qquad \PDia^h \FDia^j A \impl \FDia^i \PDia^k A.
$$

\begin{definition}
Let $\mbf S$ be a set of axiom schemes whose members are
formulae of the form $F \impl G$. 
An {\em axiomatic extension of $\SKt$} with $\mbf S$ is the
proof system obtained by adding to $\SKt$ the inference rule
$$
\infer[]
{\overline F, G}{}
$$
for each $F \impl G \in \mbf S.$
We denote with $\SKtExt{AxS}$ the axiomatic extension of $\SKt$
with axioms $\mbf S.$
\end{definition}

In the following, we shall use the notation $\fseqn n {\Delta}$
to denote the sequent
$$
\circ \underbrace{\{ \cdots \circ \{ }_n \Delta \} \cdots \}.
$$
The notation $\pseqn n \Delta$ is defined similarly. 

\begin{lemma}
\label{lm:primitive-scott-lemmon}
For any $h,i,j,k \geq 0$, the axiomatic extension of $\SKt$ with $G(h,i,j,k)$ is equivalent
to the axiomatic extension of $\SKt$ with $P(h,i,j,k).$ 
\end{lemma}
\begin{proof}
We give a syntactic proof of this lemma, i.e., we show
that the axiom $G(h,i,j,k)$ is derivable in $\SKt$
extended with axiom $P(h,i,j,k)$, and vice versa.
The axiom rules corresponding to $G(h,i,j,k)$ and
$P(h,i,j,k)$ are, respectively, 
$$
\infer[SL]
{\FBox^h \FDia^i \overline A, \FBox^j \FDia^k A}
{}
\qquad \hbox{ and }
\qquad
\infer[PSL.]
{\PBox^h \FBox^j \overline A, \FDia^i \PDia^k A}
{}
$$
In the following derivation, we make use of the fact
that deep inference rules of $\DKt$ are derivable
in $\SKt$, so we shall freely mix deep and shallow
inference rules (including residuation rules).
We shall also make use of derived rules that allow one
to go from a formula to its sequent counterpart, e.g.,
replacing $\FBox A$ with $\fseq A$, etc., which could
easily be done using appropriate cuts. So we shall also 
assume the following deep inference rules:
$$
\infer[\equiv]
{\Sigma[\fseq A]}
{\Sigma[\FBox A]}
\qquad
\infer[\equiv]
{\Sigma[\pseq A]}
{\Sigma[\PBox A]}
$$

The primitive form of Scott-Lemmon axiom can then be derived
as follows:
$$
\infer=[\PBox;\FBox]
{\PBox^h \FBox^j \overline A, \FDia^i \PDia^k A}
{
 \infer=[rf]
 {\pseqn{h}{\fseqn j {\overline A}}, \FDia^i \PDia^k A}
 {
  \infer=[\equiv]
  {\fseqn j {\overline A}, \fseqn h {\FDia^i \PDia^k A} }
  {
   \infer[cut]
   {\fseqn j {\overline A}, \FBox^h \FDia^i \PDia^k A}
   {
   \infer[SL]
   {\FBox^h \FDia^i \PDia^k A, \FBox^j \FDia^k \PBox^k \overline A}
   {}
   &
   \infer=[\FDia_1]
   {\fseqn j {\overline A}, \FDia^j \FBox^k \PDia^k A}
   {
    \infer=[\FBox]
    {\fseqn j {\overline A, \FBox^k \PDia^k A}}
    {
     \infer=[\PDia_2]
     {\fseqn j {\overline A, \fseqn k {\PDia^k A}}}
     {
      \infer[id]
      {\fseqn j {\overline A, A, \fseqn k {~} }}
      {}
     }
    }
   }
   }
  }
 }
}
$$
Note that in the derivation above, to simplify presentation, we do not keep
the principal formula of a rule in the premise as we would normally
do in $\DKt.$

It is not difficult to see that the converse also holds, i.e.,
assuming $P(h,i,j,k)$ (i.e., the rule $PSL$), one can derive
the axiom $G(h,i,j,k)$, using cuts, $rp$ and other modal/tense introduction
rules. We leave this as an exercise to the reader.
\end{proof}

Having shown the equivalence of the axioms $G(h,i,j,k)$ and $P(h,i,j,k)$,
we shall use the latter to design a cut-free extension of $\SKt$ with
Scott-Lemmon axioms. For each $P(h,i,j,k)$, we define a corresponding
structural rule as follows:
$$
\infer[sl(h,i,j,k)]
{\Gamma, \pseqn h {\fseqn j \Delta}}
{\Gamma, \fseqn i {\pseqn k \Delta}}
$$

\begin{definition}
Let $\mbf{S}$ be a set of axioms and let $\rho(\mbf{S})$ be
the corresponding structural rules for axioms in $\mbf{S}$.
The {\em structural extension} of $\SKt$ with $\mbf{S}$ is the proof system 
obtained by adding the structural rules $\rho(\mbf S)$ to $\SKt.$
We denote with $\SKtExt{S}$ the structural extension of $\SKt$ with $\mbf S.$
\end{definition}

\begin{proposition}
\label{prop:SL-axiomatic-structural}
For any set of Scott-Lemmon axioms $\mbf {SL}$, the proof systems
$\SKtExt{AxSL}$ and $\SKtExt{SL}$ are equivalent.
\end{proposition}
\begin{proof}
The following two derivations show how one can derive 
an axiom $P(h,i,j,k)$ using its structural rule counterpart, and vice
versa. 
$$
\infer=[\PBox ; \FBox]
{\PBox^h \FBox^j \overline A, \FDia^i \PDia^k A}
{
 \infer[sl(h,i,j,k)]
 {\pseqn h {\fseqn j {\overline A}}, \FDia^i \PDia^k A}
 {
  \infer[\FDia_1 ; \PDia_1]
  {\fseqn i {\pseqn k {\overline A}}, \FDia^i \PDia^k A}
  {
    \infer[id]
    {\fseqn i {\pseqn k {\overline A, A}}}
    {}
  }
 }
}
\qquad
\infer[\equiv]
{\Gamma, \pseqn h {\fseqn j \Delta}}
{
 \infer[cut]
 {\Gamma, \PBox^h \FBox^j \tau(\Delta)}
 {
  \infer[axiom]
  {\PBox^h \FBox^j \tau(\Delta), \FDia^i \PDia^k \overline{\tau(\Delta)}}{}
  &
  \deduce
  {\Gamma, \FBox^i \PBox^k \tau(\Delta)}
  {\deduce{\vdots}{\Gamma, \fseqn i {\pseqn k \Delta}}}
 }
}
$$
Note that in presenting the derivations, we adopt the
same simplication steps as we did in Lemma~\ref{lm:primitive-scott-lemmon}.
\end{proof}

As noted earlier, $\SKtExt{AxSL}$ does not have cut elimination,
as typical for axiomatic extensions of sequent calculi, although one
could perhaps show that applications of the cut rule can be limited to 
those that cut directly with the axioms. 
But we shall show that the ``pure'' sequent calculus $\SKtExt{SL}$
does enjoy true cut elimination. This is a simple consequence
of Theorem~\ref{thm:cut-elim}, as the rules in $\rho(\mbf{SL})$ are 
substitution-closed linear rules. 

\begin{theorem}
\label{thm:ce-SKtExt}
For any set of Scott-Lemmon axioms $\mbf{SL}$, 
cut elimination holds for $\SKtExt{SL}$.
\end{theorem}

In the following subsections, we consider three instances
of $\SKtExt{SL}$, i.e., extensions of $\SKt$ with axioms
for $\mathit{S4}$, $\mathit{S5}$, and the axiom of uniqueness. 
We give deep inference systems for these
logics that are equivalent to their shallow counterparts. These are by no means
an exhaustive list of logics for which the correspondence between
deep and shallow systems holds; they are meant as an illustration of the kind of 
methods used to eliminate structural rules via propagation rules.
For the extensions with $\mathit{S4}$ and $\mathit{S5}$, the proofs
of the correspondence are not very different from the proof of the correspondence 
between $\SKt$ and $\DKt$, so we shall only state the correspondence results
and omit the proofs. 
The interested reader can consult the doctoral thesis of the second author~\cite{Postniece10Phd} 
for details. We shall present a more general framework in Section~\ref{sec:path},
in which this correspondence can be proved uniformly for a class of
axiomatic extensions of $\SKt.$ 

\subsection{A deep-inference system for modal tense logic KtS4} 

\begin{figure}[t]
$$
\begin{array}{c@{\qquad\quad}c@{\qquad\quad}c}
\infer[T_a]
{\Sigma[\PDia A]}
{\Sigma[\PDia A, A]}
&
\infer[4_a]
{\Sigma[\PDia A, \pseq{\Delta}]}
{\Sigma[\PDia A, \pseq{\PDia A, \Delta}]}
&
\infer[4_c]
{\Sigma[\FDia A, \fseq{\Delta}]}
{\Sigma[\FDia A, \fseq{\FDia A, \Delta}]}
\\ \\
\infer[T_b]
{\Sigma[\FDia A]}
{\Sigma[\FDia A, A]}
&
\infer[4_b]
{\Sigma[\fseq{\Delta, \PDia A}]}
{\Sigma[\fseq{\Delta, \PDia A}, \PDia A]}
&
\infer[4_d]
{\Sigma[\pseq{\Delta, \FDia A}]}
{\Sigma[\pseq{\Delta, \FDia A}, \FDia A]}
\end{array}
$$
\caption{Additional propagation rules for $\DSFour$}
\label{fig:dsfour}
\end{figure}

Consider an extension of $\SKt$ with the axioms for 
reflexivity and transitivity (given in primitive form):  
$T: A \impl \FDia A$ and $4: \FDia \FDia A \impl \FDia A.$
Their corresponding structural rules are: 
$$
\infer[T_f]
{\Gamma, \Delta}
{\Gamma, \fseq \Delta}
\qquad
\infer[4_f.]
{\Gamma, \fseq {\fseq{\Delta}}}
{\Gamma, \fseq \Delta}
$$
Using residuation, we can also derived the tense counterparts
of the rule $T_f$ and $4_f$, with the structural connective $\circ$ replaced
by $\bullet$:
$$
\infer[T_f]
{\Gamma, \Delta}
{
 \infer[rf]
 {\fseq \Gamma, \Delta}
 {
  \Gamma, \pseq \Delta
 }
}
\qquad \qquad
\infer[rp]
{\Gamma, \pseq {\pseq \Delta}}
{
 \infer[rp]
 {\fseq \Gamma, \pseq \Delta}
 {
  \infer[4_f]
  {\fseq {\fseq \Gamma}, \Delta}
  {
   \infer[rf] 
   {\fseq \Gamma, \Delta}
   {\Gamma, \pseq \Delta}
  }
 }
}
$$
As with the design of $\DKt$, in designing a deep inference system
for KtS4, we aim to get rid of all structural rules.
This is achieved via propagation rules for $\FDia$-formulae,
and by residuation, also for $\PDia$-formulae. The propagation
rules needed are given in Figure~\ref{fig:dsfour}. 

\begin{definition}
We denote with $\SSFour$ the proof system obtained by
adding to $\SKt$ the structural rules $T_f$ and $4_f.$ 
System $\DSFour$ denotes $\DKt$ plus the propagation rules
given in Figure~\ref{fig:dsfour}. 
\end{definition}
The purely modal rules of $\DSFour$, i.e.,$T_b$ and $4_c$, 
coincide with Br\"unnler's rules for $T$ and $4$ in \cite{brunnler2006}.
The rules of $\DSFour$ can be shown to be derivable in $\SSFour$.

\begin{theorem}
For every $\Gamma$, we have $\vdash_\SSFour \Gamma$ if and only if 
$\vdash_\DSFour \Gamma.$
\end{theorem}

As with $\DKt$, if we restrict $\DSFour$ to its purely modal fragment, we
obtain a sound and complete proof system for modal logic S4.
Let $\DKFour$ be $\DK$ extended with $T_b$ and $4_c.$
The proof of the following theorem is similar to the proof of Theorem~\ref{thm:separation}.

\begin{theorem}[Separation]
For every modal formula $A$,  
$\vdash_{\DKFour} A$ iff $A$  is a theorem of S4.
\end{theorem}

\subsection{A deep-inference system for modal tense logic S5}

\begin{figure}[t]
$$
\begin{array}{c@{\quad}c@{\quad}c@{\quad}c}
\infer[5_a]
{\Sigma[\PDia A, \fseq {\Delta}]}
{\Sigma[\PDia A, \fseq {\PDia A, \Delta}]}
&
\infer[5_b]
{\Sigma[\fseq {\Delta, \FDia A}]}
{\Sigma[\fseq {\Delta, \FDia A}, \FDia A]}
&
\infer[5_c]
{\Sigma[\FDia A, \pseq {\Delta}]}
{\Sigma[\FDia A, \pseq {\FDia A, \Delta}]}
&
\infer[5_d]
{\Sigma[\pseq {\Delta, \PDia A}]}
{\Sigma[\pseq {\Delta, \PDia A}, \PDia A]}
\end{array}
$$
\caption{Additional propagation rules for $\DSFive$}
\label{fig:dsfive}
\end{figure}

We can obtain KtS5 from $\SSFour$ by adding the symmetry axiom $B: A \impl \FBox \FDia A.$ 
The corresponding primitive form of $B$ is $\FDia A \impl \PDia A$, and
its corresponding structural rule is 
$$
\infer[B]
{\Gamma, \fseq \Delta}
{\Gamma, \pseq \Delta}
$$
The additional propagation rules, on top of those for $\DSFour$, 
needed to absorb this structural rule and those
of $\SSFour$ are given in Figure~\ref{fig:dsfive}. 

\begin{definition} System $\SSFive$ is 
 $\SSFour$ plus the rule $B.$
System $\DSFive$ is  $\DSFour$ plus the propagation rules given in
Figure~\ref{fig:dsfive}.
\end{definition}

Note that as a consequence of symmetry, the forward-looking and the backward-looking
modal operators (and their structural counterparts) collapse. Hence, the propagation of 
diamond-formulae becomes `colour-blind', 
i.e., $\FDia$ behaves exactly as $\PDia$ in any context. 
This simplifies significantly the proof of admissibility of
structural rules of $\SSFive$ in $\DSFive$, in particular, admissibility of $B$. 

\begin{theorem}
For every $\Gamma$, we have $\vdash_\SSFive \Gamma$ if and only if $\vdash_\DSFive \Gamma.$
\end{theorem}

Note that $\DSFive$ captures $S5 = KT4B$ rather than $S5 = KT45$. It is also possible
to formulate deep inference rules that correspond directly to axiom $5$, but
one would need a form of global propagation rule (see Section~\ref{sec:path}). 
Again, as with $\DSFour$, the separation property also holds for $\DSFive$.
Let $\DKFive$ be the restriction of $\DSFive$ to the purely modal fragment.

\begin{theorem}[Separation]
For every modal formula $A$,  
$\vdash_{\DKFive} A$ iff $A$  is a theorem of S5.
\end{theorem}

\subsection{A deep inference system for an extension of Kt with the axiom of uniqueness}

We now consider extending Kt with the axiom $CD : \FDia A \impl \FBox A$.
Its primitive form is $\PDia \FDia A \impl A$ and its corresponding structural rule is
$$
\infer[U.]
{\Gamma, \pseq {\fseq \Delta}}
{\Gamma, \Delta}
$$
The propagation rules needed to absorb this structural rules are as follows:
$$
\infer[u_1]
{\Sigma[A, \pseq {\Gamma, \fseq \Delta}]}
{\Sigma[A, \pseq {\Gamma, \fseq {A, \Delta}}]}
\qquad
\infer[u_2]
{\Sigma[\fseq {\Delta_1, A}, \fseq{\Delta_2}]}
{\Sigma[\fseq {\Delta_1, A}, \fseq{A, \Delta_2}]}
\qquad
\infer[u_3]
{\Sigma[\pseq {\Gamma, \fseq {A, \Delta}}]}
{\Sigma[A, \pseq {\Gamma, \fseq {A, \Delta}}]}
$$
\begin{definition} System $\SSU$ is 
 $\SKt$ plus the rule $U.$
System $\DKtU$ is  $\DKt$ plus the propagation rules $u_1, u_2$ and $u_3.$ 
\end{definition}

\begin{lemma}
\label{lm:DKtU-SSU}
Every rule of $\DKtU$ is derivable in $\SSU.$
\end{lemma}
\begin{proof}
Since all the rules of $\DKt$ are derivable in $\SKt$, which is a subset of $\SSU$,
it is enough to show that the additional propagation rules $u_1,u_2$ and $u_3$
are derivable in $\SSU.$ Figure~\ref{fig:DSU} shows the derivations of $u_1$ 
(the left figure) and $u_2$ (the right figure). The rule $u_3$ can be derived
similarly, i.e., using $u_1$ and appropriate applications of residuation. 
In the derivation of $u_1$, we use implicitly Proposition~\ref{prop:display}
to display nested structures, and the fact that deep inference rules $\FDia_1$
and $\PDia_1$, and the deep weakening rule are derivable in $\SKt$.
\end{proof}

\begin{theorem}
For every $\Gamma$, we have $\vdash_\SSU \Gamma$ if and only if $\vdash_\DKtU \Gamma.$
\end{theorem}
\begin{proof}
Lemma~\ref{lm:DKtU-SSU} shows one direction; it remains to show the other, i.e.,
that every cut-free derivation of $\SSU$ can be transformed into a derivation in $\DKtU.$
As with the case with $\DSFour$ and $\DSFive$, we need to first prove admissibility
of all structural rules. This can be done by straightforward induction on the height
of derivations and case analyses on the last rules of the derivations. There 
are numerous tedious cases to consider, but none are difficult; we leave them
as an exercise for the reader. 
\end{proof}

By restricting to the purely modal fragment of $\DKtU$, we get a sound and complete
proof system for modal logic $K + \mathit{CD}$.
Let $\DKU$ be the modal fragment of $\DKtU$, i.e., $\DK$ plus
the rule $u_2.$ 
\begin{theorem}[Separation]
For every modal formula $A$,  
$\vdash_{\DKU} A$ iff $A$  is a theorem of the modal logic $K + \mathit{CD}$.
\end{theorem}

\begin{figure}[t]
$$
\begin{array}{cc}
\infer=[rf;rp]
{\Sigma[A, \pseq{\Gamma, \fseq{\Delta}}]}
{
 \infer[ctr]
 {\Psi, A, \pseq{\Gamma, \fseq{\Delta}}}
 {
  \infer[cut]
  {\Psi, A, A, \pseq{\Gamma, \fseq{\Delta}}}
  {
   \infer=[\PBox;\FBox]
   {\PBox \FBox \overline A, A}
   {
     \infer[U]
           {\pseq {\fseq {\overline A}}, A}
           {
             \infer[id]
                   {\overline A, A}
                   {}
           }
   }
   &
   \infer=[\PDia_1 ; wk]
   {\Psi, A, \PDia \FDia A, \pseq{\Gamma, \fseq{\Delta}}}
   {
    \infer=[\FDia_1 ; wk]
    {\Psi, A,  \pseq{\Gamma, \FDia A, \fseq{\Delta}}}
    {
      \infer=[rf;rp]
      {\Psi, A,  \pseq{\Gamma, \fseq{\Delta, A}}}
      {\Sigma[A, \pseq{\Gamma, \fseq{A, \Delta}}]}
    }
   }
  }
 }
} 
& 
\qquad
\infer=[rp;rf]
{\Sigma[\fseq {\Delta_1,A}, \fseq{\Delta_2}]}
{
 \infer[rp]
 {\Psi, \fseq{\Delta_1,A}, \fseq{\Delta_2}}
 {
  \infer[u_1]
  {\Delta_1, A, \pseq{\Psi, \fseq{\Delta_2}}}
  {
    \infer[rf]
    {\Delta_1, A, \pseq{\Psi, \fseq{A,\Delta_2}}}
    {
     \infer=[rp;rf]
     {\Psi, \fseq{\Delta_1,A}, \fseq{A, \Delta_2}}
     {\Sigma[\fseq{\Delta_1,A}, \fseq{A, \Delta_2}]}
    }
  }
 }
}
\\ \\
(1) & \qquad (2)
\end{array}
$$
\caption{Derivations of the rules $u_1$ and $u_2$.}
\label{fig:DSU}
\end{figure}

\section{Path axioms and global propagation rules}
\label{sec:path}

We now consider extensions of $\Kt$ with a class of 
axioms which we call {\em path axioms}. As the name suggests,
these axioms can be seen as describing paths in a tree
of sequents along which formulae can propagate.
We show that $\Kt$ extended with path axioms can be
formulated in both the shallow calculus and the
deep calculus. For the latter, the formulation of
the propagation rules is derived naturally from 
the (transitive closure of) axioms.

Before we proceed, it will be helpful to draw a distinction
between a formula and a {\em schematic formula}. 
We have so far blurred this distinction when we discuss
axioms (which are schematic formulae) and their
instances. 
By a schematic formula, we mean syntactic expressions composed
using logical connectives and {\em meta variables}.
We shall denote meta variables with $X,Y$ and $Z.$
A formula scheme can be instantiated by substituting
its meta variables with (concrete) formulae or
other formulae schemes. By axioms, we usually mean
schematic formulae whose (concrete) instances are admitted
as theorems of the logic. In the following, 
we shall make explicit this distinction between formulae
and schematic formulae. We shall also use the notation $\udia$
(possibly with subscripts) to denote a diamond-operator of either color, and $\ubox$
to denote its de Morgan dual. 
 
\begin{definition}
A {\em path axiom} is a schematic formula for the form
$
\udia_1 \cdots \udia_{n} X
\impl 
\udia X
$
where $n \geq 0$, and each of $\{\udia, \udia_{1}, \ldots, \udia_{n}\}$ 
is either a $\FDia$, or a $\PDia.$
\end{definition}
The class of path axioms includes any instance of primitive Scott-Lemmon axiom
$P(h,i,j,k)$ where $i+k = 1$. By Lemma~\ref{lm:primitive-scott-lemmon},
these are equivalent to the following instances of Scott-Lemmon axioms:
$$
\FDia^h \FBox X \impl \FBox^j X
\qquad
\FDia^h X \impl \FBox^j \FDia X.
$$
Hence, it subsumes most standard axioms such as 
reflexivity ($\FBox X \impl X$), transitivity ($\FDia \FDia X \impl \FDia X$),
symmetry ($X \impl \FBox \FDia X$), and euclideanness ($\FDia X \impl \FBox \FDia X$).

To each path axiom, $\udia_1 \cdots \udia_n X \impl \udia X $, we define a 
corresponding structural rule as shown below left, where $\star$ is the structural connective for $\ubox$ 
and each $\star_i$ is the structural connective for $\ubox_i.$ 
For example, the structural rule for the axiom $\FDia \PDia \FDia X \impl \FDia X$
is as given below right. 
$$
\infer[\rho]
{\Gamma, \star_1 \{ \cdots \star_n \{ \Delta \} \cdots \} }
{\Gamma, \star \{ \Delta \} }
\qquad
\infer[ .]
{\Gamma, \fseq {\pseq {\fseq \Delta}}}
{\Gamma, \fseq \Delta}
$$
Given a set of axioms $\mbf P$, we denote with $\rho(\mbf P)$ the set of structural
rules corresponding to axioms in $\mbf P.$ As with Scott-Lemmon axioms,
axiomatic and structural extensions of $\SKt$ with path axioms are equivalent.
The proof of the following proposition is similar to the proof of
Proposition~\ref{prop:SL-axiomatic-structural}.
\begin{proposition}
For any set $\mbf P$ of path axioms, the proof systems $\SKtExt{AxP}$ and
$\SKtExt{P}$ are equivalent.
\end{proposition}
As a corollary of Theorem~\ref{thm:cut-elim}, cut elimination holds for
$\SKtExt{P}.$ 
\begin{theorem}
Cut elimination holds for $\SKtExt{P}$, for any set $\mbf P$ of
path axioms.
\end{theorem}

\subsection{Propagation rules for path axioms}

A straightforward way to incorporate a path axiom, say, $\FDia \PDia X \impl \FDia X$
in the deep inference system $\DKt$ is to simply use it as a rule, by replacing $\FDia \PDia X$ with $\FDia X$ (reading
the rule top down), i.e., 
$$
\infer[ .]
{\Sigma[\FDia X]}
{\Sigma[\FDia \PDia X]}
$$
Despite its appealing simplicity, adding such a rule will destroy the subformula property,
and as our main goal is to design proof-search friendly calculi, such an
introduction rule must be ruled out. What we propose here is essentially
the same, but instead of putting the formula $\FDia \PDia X$ in the premise,
we consider all its possible interactions with the surrounding context ($\Sigma[~]$)
to decompose it to $X.$ This would involve propagating $X$ to different 
subcontexts in $\Sigma[~]$, depending on the axiom. The main challenge
here is then to design a sound and complete set of propagation rules
for the axiom. 

To understand the intuitive idea behind propagation rules for
path axioms, it is helpful to view a nested sequent as a tree
of traditional sequents. Following Kashima~\cite{kashima-cut-free-tense},
we define a mapping from sequents to trees as follows.
A \emph{node} is a multiset of formulae. A \emph{tree} is a node with 0 or
more children, where each child is a tree, and each child is labelled
as either a $\circ$-child, or a $\bullet$-child.  Given a sequent $\Xi
= \Theta, \fseq {\Gamma_1}, \cdots, \fseq {\Gamma_n}, \pseq
{\Delta_1}, \cdots, \pseq {\Delta_m}$, where $\Theta$ is a multiset of
formulae and $n \geq 0$ and $m \geq 0$, the tree $tree(\Xi)$
represented by $\Xi$ is:
$$
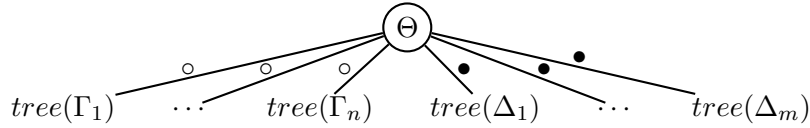
\mbox{
\psset{arrows=-}
\pstree[levelsep=30pt,labelsep=0pt]{\Tcircle{$\Theta$}}
{
\TR{$tree(\Gamma_1)$}\tlput{$\circ$}
\TR{$\cdots$}\tlput{$\circ$}
\TR{$tree(\Gamma_n)$}\tlput{$\circ$}
\TR{$tree(\Delta_1)$}\trput{$\bullet$}
\TR{$\cdots$}\trput{$\bullet$}
\TR{$tree(\Delta_m)$}\taput{$\bullet$}
}
}
$$
In $\DKt$, a $\FDia$- or a $\PDia$-prefixed formula can navigate up and down
a sequent tree, depending on where it is positioned in the tree.
The rule $\FDia_1$ allows a formula $\FDia A$
to propagate its subformula $A$ down the tree along an edge
labelled by $\circ$, and $\FDia_2$ allows the same formula to propagate
$A$ up the tree along an edge labelled by $\bullet.$ Similarly,
$\PDia_1$ allows $\PDia A$ to propagate $A$ down an $\bullet$-edge
and $\PDia_2$ allows it to propagate $A$ up an $\circ$-edge.
Graphically, one can represent these movements by assigning two kinds of
diamond-labelled directed edges to each edge in a sequent tree, which encode
the kinds of diamond-prefixed formulae that can propagate along the
directed edges. The four movements mentioned previously can thus
be represented as the dotted lines in the following graph:
$$
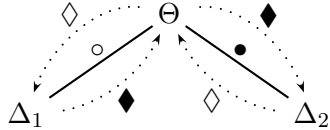
\mbox{
\pstree[levelsep=8ex,treesep=20ex,nodesep=2pt,labelsep=0pt]
{\TR[name=R]{$\Theta$}}
{
 
 \TR[name=C1]{$\Delta_1$} \taput{$\circ$}
 \TR[name=C2]{$\Delta_2$} \taput{$\bullet$}
}
\psset{nodesep=4pt,arrows=->,linestyle=dotted,labelsep=0pt}
\ncarc[arcangle=-30]{C1}{R} \nbput{$\PDia$}
\ncarc[arcangleA=-50,arcangleB=-35]{R}{C1} \nbput{$\FDia$}
\ncarc[arcangleA=50,arcangleB=35]{R}{C2} \naput{$\PDia$}
\ncarc[arcangle=30]{C2}{R} \naput{$\FDia$}
}
$$
For example, the ``diamond paths'' from the node labelled by $\Delta_1$ to
$\Delta_2$ characterise the diamond prefixes needed to propagate a formula
from $\Delta_1$ to $\Delta_2$; they include formulae such as 
$\PDia\PDia A$ (one goes up to the root and then down to $\Delta_2$),
or $\PDia \FDia \PDia \PDia A$ (i.e., one does a ``loop'' from $\Delta_1$
to $\Theta$ and back to $\Delta_1$, before proceeding to $\Delta_2$), etc. 

In proof search, a path axiom such as $\FDia \PDia \FDia X \impl \FDia X$ can
be read as an instruction for propagating a formula $\FDia A$: replace $\FDia A$
with $\FDia\PDia \FDia A$ and propagate along the diamond path $\FDia\PDia\FDia.$
Depending on where the formula $\FDia A$ is located in a sequent tree, there
are several possible moves that correspond to the path $\FDia \PDia \FDia.$
Some of these are given in Figure~\ref{fig:propagation}.

\begin{figure}[t]
$$
\mbox{
\pstree[levelsep=8ex,treesep=15ex,nodesep=2pt,labelsep=0pt]
{\TR[name=R1]{$\Theta, \FDia A$}}
{
 \Tfan
 \pstree
 {\TR[name=D11]{$\Delta_1$}_{$\circ$} }
 {
  \Tfan
  \pstree
  {\TR[name=D12]{$\Delta_2$}_{$\bullet$}}
  {
   \Tfan
   \TR[name=D13]{$\Delta_3, A$}_{$\circ$}
  }
 }
}
\psset{nodesep=4pt,arrows=->,linestyle=dotted,labelsep=0pt}
\ncarc[arcangle=40]{R1}{D11} \naput{$1:\FDia$}
\ncarc[arcangle=40]{D11}{D12} \naput{$2:\PDia$}
\ncarc[arcangle=40]{D12}{D13} \naput{$3:\FDia$}
}
\qquad
\mbox{
\pstree[levelsep=8ex,treesep=15ex,nodesep=2pt,labelsep=0pt]
{\TR[name=R2]{$\Theta, \FDia A$}}
{
 \Tfan
 \pstree
 {\TR[name=D21]{$\Delta_1,A$}_{$\circ$} }
 {
  \Tfan
  \TR[name=D22]{$\Delta_2$}_{$\bullet$}
 }
}
\psset{nodesep=4pt,arrows=->,linestyle=dotted,labelsep=0pt}
\ncarc[arcangle=40]{R2}{D21} \naput{$1:\FDia$}
\ncarc[arcangle=40]{D21}{D22} \naput{$2:\PDia$}
\ncarc[arcangle=40]{D22}{D21} \naput{$3:\FDia$}
}
$$

$$
\mbox{
\pstree[levelsep=8ex,treesep=15ex,nodesep=2pt,labelsep=0pt]
{\TR[name=R3]{$\Theta$}}
{
 \Tfan
 \pstree
 {\TR[name=D31]{$\Delta_1, \FDia A$}_{$\bullet$} }
 {
  \Tfan
  \TR[name=D32]{$\Delta_2, A$}_{$\circ$}
 }
}
\psset{nodesep=4pt,arrows=->,linestyle=dotted,labelsep=0pt}
\ncarc[arcangle=40]{D31}{R3} \naput{$1:\FDia$}
\ncarc[arcangle=40]{R3}{D31} \naput{$2:\PDia$}
\ncarc[arcangle=40]{D31}{D32} \naput{$3:\FDia$}
}
\qquad
\mbox{
\pstree[levelsep=8ex,treesep=20ex,nodesep=2pt,labelsep=0pt]
{\TR[name=R4]{$\Theta, A$}}
{
 \Tfan
 \TR[name=D4]{$\Delta, \FDia A$} \taput{$\bullet$}
}
\psset{nodesep=4pt,arrows=->,linestyle=dotted,labelsep=0pt}
\ncarc[arcangle=-30]{D4}{R4} \nbput{$1:\FDia$}
\ncarc[arcangle=-30]{R4}{D4} \nbput{$2:\PDia$}
\nccurve[angle=30,ncurv=1]{D4}{R4} \nbput{$3:\FDia$}
}
$$
\caption{Some propagation scenarios for axiom $\FDia \PDia \FDia A \impl \FDia A.$}
\label{fig:propagation}
\end{figure}

In designing the propagation rules for a set of path axioms, in order
to get completeness, one needs to take into account two things: 
arbitrary compositions of the axioms and their interactions with
the residuation axioms. An axiom such as $\FDia \PDia \FDia X \impl \FDia X$
not only specifies a set of possible propagations for $\FDia A$,
but also specifies, via residuation, propagations for $\PDia A$. 
It is easy in this case to show that
$\PDia \FDia \PDia X \impl \PDia X$ is a consequence of that axiom.  

In the following, when $\udia$ denotes an diamond operator ($\PDia$ or $\FDia$),
$\udia^{-1}$ denotes its tense or modal counterpart. That is,
if $\udia = \FDia$ then $\udia^{-1}$ denotes $\PDia$ and vice versa. 

\begin{definition}
Let $F$ be the path axiom $\udia_1 \cdots \udia_n X \impl \udia~X.$ 
The {\em inverted version}
of $F$, denoted by $I(F)$,
is the schematic formula 
$\udia_n^{-1} \cdots \udia_1^{-1} X \impl \udia^{-1} X.$
\end{definition}
Obviously, we have $I(I(F)) = F.$ 
A path axiom can be shown equivalent to its inverted version. 

\begin{lemma}
\label{lm:inverted-equiv}
Let $F$ be a path axiom. Then $F$ is equivalent to $I(F)$. 
\end{lemma}
\proof
Since $I(I(F)) = F$ and $I(F)$ itself is a path axiom, it is enough to show
one direction, i.e., $F$ implies $I(F).$
We first note that the following are theorems of tense logic (they are,
in fact, the axioms of residuation): 
$$
X \impl \FBox \PDia X
\qquad
X \impl \PBox \FDia X.
$$
There are two cases to consider:
\begin{enumerate}[$\bullet$]
\item $F = \udia_1 \cdots \udia_n X \impl \FDia X.$
Then $I(F) = \udia_n^{-1} \cdots \udia_1^{-1} X \impl \PDia X.$
By contrapositon, we have that $F$ implies $\FBox X \impl \ubox_1 \cdots \ubox_n X$. 
By instantiating this axiom scheme with $\PDia X$, we have
$
\FBox \PDia X \impl \ubox_1 \cdots \ubox_n \PDia X.
$
Since $X \impl \FBox \PDia X$, we also have 
$
X \impl \ubox_1 \cdots \ubox_n \PDia X.
$
Note that since $\ubox_i$ is the de Morgan dual of $\udia_i$,
its residual must be $\udia_i^{-1}$. Therefore, 
by residuation, we have 
$$
\udia_n^{-1} \cdots \udia_1^{-1} X \impl \PDia X.
$$

\item $F = \udia_1 \cdots \udia_n X \impl \PDia X.$ This is similar
to the previous case, except that we compose with the
axiom $X \impl \PBox \FDia X$.\qed
\end{enumerate}

\begin{definition}
\label{def:composition}
Let $F$ and $G$ be the following path axioms:
$$
\udia_{F_1} \cdots \udia_{F_m} X \impl \udia_F X
\qquad \qquad
\udia_{G_1} \cdots \udia_{G_n} X \impl \udia_G X.
$$
$F$ is said to be {\em composable with $G$ at position $i$}
if $\udia_F = \udia_{G_i}.$ We denote by $F \triangleright^i G$
the composition of $F$ with $G$ at $i$, i.e., the formula: 
$$
\udia_{G_1} \cdots \udia_{G_{i-1}} \udia_{F_1} \cdots \udia_{F_m} \udia_{G_{i+1}} \cdots \udia_{G_n} \impl \udia_G X.
$$
We say that {\em $F$ is composable with $G$} if $F$ is composable with $G$
at some position $i.$
We denote with $F \triangleright G$ the set of all compositions
of $F$ with $G$, i.e., 
$$
F \triangleright G = \{F \triangleright^i G \mid \mbox{$F$ composable with $G$ at $i$} \}. 
$$
\end{definition}
Notice that composition of axioms are basically just modus ponens, 
so the compositions of $F$ and $G$ are obviously logical consequences of $F$ and $G$. 

\begin{lemma}
\label{lm:comp-equiv}
If $F$ is composable with $G$ at $i$, then $F\triangleright^i G$ is 
a logical consequence of $F$ and $G.$
\end{lemma}

\begin{definition}
\label{def:completion}
Let $\mbf{P}$ be a set of path axioms. The {\em completion of $\mbf{P}$},
written $\mbf{P}^*$, is the smallest set of path axioms containing $\mbf{P}$ 
and satisfying the following conditions:
\begin{enumerate}
\item It contains the identity axioms $\FDia X \impl \FDia X$ and $\PDia X \impl \PDia X.$
\item It is closed under {\em composition}, i.e., if $F, G \in \mbf{P}^*$
and $F$ is composable with $G$, then $F \triangleright G \subseteq \mbf{P}^*.$
\end{enumerate}
\end{definition}

Alternatively, we can characterise $\mbf{P}^*$ via
a monotone operator: 
$$
\Ccal(S) = \bigcup \{ F \triangleright G \mid F, G \in S \mbox{ and $F$ is composable with $G$} \}.
$$
Now define an $n$-th iteration of $\Ccal$ as follows:
$$
\begin{array}{l}
\Ccal^0(S) = \emptyset \\
\Ccal^{n+1}(S) = S \cup \Ccal(\Ccal^{n}(S)).
\end{array}
$$
Then it can be shown that (see \cite{aczel-handbook})
$$
\mbf{P}^* = \bigcup_{n < \omega} 
\Ccal^n(\mbf P \cup \{\FDia X \impl \FDia X, \PDia X \impl \PDia X\}).
$$
That is, every element of the set $\mbf{P}^*$
can be obtained via a finite number of compositions using axioms in the
set $\mbf P \cup \{\FDia X \impl \FDia X, \PDia X \impl \PDia X\}.$
We shall use this fact in the proofs involving the completion of $\mbf P.$

In the following, we lift the operator $I$ to a set of axioms,
i.e., $I(\mbf{P}) = \{I(F) \mid F \in \mbf{P} \}.$

\begin{lemma}
\label{lm:completion}
Let $\mbf{P}$ be a set of path axioms. 
If $I(\mbf P) \subseteq \mbf P$ then for every $F \in \mbf{P}^*$
we have $I(F) \in \mbf{P}^*.$
\end{lemma}
\begin{proof}
By induction on the formation of the set $\mbf{P}^*$ and Definition~\ref{def:composition}.
\end{proof}

To define the propagation rules, we need to define 
the notion of a path between two nodes in a tree. This is given
in the following.

\begin{definition}
Let $\Gamma$ be a nested sequent, and let $N$ be the set
of nodes of $tree(\Gamma).$ The {\em propagation graph $PG(\Gamma)$ 
for $\Gamma$} is a directed graph such that
the set of nodes of $PG(\Gamma)$ is $N$, its edges are labelled with $\PDia$
or $\FDia$ and are defined as follows:
\begin{enumerate}[$\bullet$]
\item For each node $n \in N$, and each $\circ$-child $n_1$
of $n$, there is exactly one edge $(n, n_1)$ labelled with $\FDia$,
and exactly one edge $(n_1, n)$ labelled with $\PDia.$
\item For each node $n \in N$, and each $\bullet$-child $n_1$
of $n$, there is exactly one edge $(n, n_1)$ labelled with $\PDia$,
and exactly one edge $(n_1, n)$ labelled with $\FDia.$
\end{enumerate}
A {\em labelled path} (or simply, a path) in a propagation graph is defined as usual,
i.e., as a sequence of nodes and diamonds, separated by semicolons,  
$$
n_1 ; \udia_1 ; n_2 ; \udia_2 ; \cdots ; n_{k-1} ; \udia_{k-1} ; n_k
$$
such that each $(n_i, n_{i+1})$ is a $\udia_i$-labelled edge in
$PG(\Gamma).$
We use $\pi$ to range over paths in a propagation graph. If $\pi$ is a path then
$\seqdia \pi$ denotes the sequence of labels (i.e., $\PDia$
or $\FDia$) that occur along that path. 
\end{definition}

We are now ready to define the set of propagation rules for
a set of axioms. But first we introduce a notational convention
for writing contexts. Note that since a context is just a
structure with a hole $[]$ in place of a formula, it also
has a tree representation. In a single-hole context,
the hole $[]$ occupies a unique node in the tree.
We shall write $\Sigma[]_i$ when we want to be explicit
about the particular node $i$ where the hole is located. 
This notation extends to multiple-hole contexts, e.g.,
$\Sigma[]_i[]_j$ denotes a two-hole context where the first
hole is located at node $i$ and the second at node $j$ in
$tree(\Sigma[][]).$

\begin{definition}
\label{def:prop-rules}
Let $\mbf{P}$ be a set of path axioms. The set of
{\em propagation rules for $\mbf{P}$}, written $Prop(\mbf P)$,
consists of rules of the form:
$$
\infer[]
{\Sigma[\udia A]_i [\emptyset]_j}
{\Sigma[\udia A]_i [A]_j}
$$
if there is a path $\pi$ from $i$ to $j$ in $PG(\Gamma)$ 
such that $\seqdia{\pi} X \impl \udia X \in (\mbf P \cup I(\mbf P))^*.$

We denote with $\SKtExt{P}$ the structural extension of 
$\SKt$ with $\mbf P$ and $\DKtExt{P}$ the extension of 
$\DKt$ with propagation rules $Prop(\mbf P).$
\end{definition}

Notice that by definition, the rule $\FDia_1$, $\FDia_2$,
$\PDia_1$ and $\PDia_2$ are just instances of propagation rules,
i.e., they are propagation rules for the identity
axiom $\FDia X \impl \FDia X$ and $\PDia X \impl \PDia X.$
So in the following proofs, we do not explicitly do case
analyses on instances of these rules, as they are subsumed
by the more general cases involving the propagation rules. 

\begin{lemma}
\label{lm:DKtP-SKtP}
For any set of path axioms $\mbf{P}$ and any structure $\Gamma$,
if $\vdash_{\DKtExt{P}} \Gamma$ then $\vdash_{\SKtExt{P}} \Gamma.$
\end{lemma}
\begin{proof}
Since $\SKt$ is a subset of $\SKtExt{P}$, derivations of $\DKt$ rules 
in $\SKtExt{P}$ are done as in Theorem~\ref{thm:SKt-equal-DKt}.
It remains to show derivations of the propagation rules.
It is enough to show that each instance of each axiom in $(\mbf P \cup I(\mbf P))^*$
is derivable in $\SKtExt{P}.$ This in effect would allow us 
to derive the following rule (via cut and Proposition~\ref{prop:display}):
$$
\infer[]
{\Sigma[\udia A]}
{
 \Sigma[\udia_1 \cdots \udia_n A]
}
$$
for each axiom $\udia_1 \cdots \udia_n A \impl \udia~A$, which would
then allow us to mimick the propagation rule for that axiom.   
The derivation of the axioms of $(\mbf P \cup I(\mbf P))^*$
follows straightforwardly from Lemma~\ref{lm:inverted-equiv}, Lemma~\ref{lm:comp-equiv}
Definition~\ref{def:completion} and Lemma~\ref{lm:completion}.
\end{proof}

\begin{figure}[t]
$$
\mbox{
\pstree[levelsep=8ex,treesep=8ex,nodesep=2pt,labelsep=0pt]
{\TR[name=R]{$\cdot$}}
{
 \pstree[linestyle=none,levelsep=2ex,arrows=-]{\Tfan[fansize=10ex]}{\TR[name=C2]{$\Gamma$}}
 \pstree{\TR[name=C1]{$\cdot$} \trput{$\bullet$}}
        {\pstree[linestyle=none,levelsep=2ex]{\Tfan[fansize=10ex]}{\TR[name=C3]{$\Delta$}}}
}
\psset{nodesep=4pt,arrows=->,linestyle=dotted,labelsep=0pt}
\ncarc[arcangle=45]{R}{C1} \naput{$\PDia$}
\ncarc[arcangle=45]{C1}{R} \naput{$\FDia$}
}
\qquad
\raisebox{-10ex}{$\Longrightarrow$}
\qquad
\mbox{
\pstree[levelsep=8ex,treesep=8ex,nodesep=2pt,labelsep=0pt]
{\TR[name=R1]{$\cdot$}}
{
 \pstree{\TR[name=D1]{$\cdot$} \tlput{$\circ$}}
        {\pstree[linestyle=none,levelsep=2ex]{\Tfan[fansize=10ex]}{\TR[name=D2]{$\Gamma$}}}
 \pstree[linestyle=none,levelsep=2ex,arrows=-]{\Tfan[fansize=10ex]}{\TR[name=D3]{$\Delta$}}
}
\psset{nodesep=4pt,arrows=->,linestyle=dotted,labelsep=0pt}
\ncarc[arcangle=45]{D1}{R1} \naput{$\PDia$}
\ncarc[arcangle=45]{R1}{D1} \naput{$\FDia$}
}
$$

$$
\mbox{
\pstree[levelsep=10ex,treesep=8ex,nodesep=2pt,labelsep=0pt]
{\TR[name=R]{$\cdot$}}
{
 \pstree{\TR[name=C1]{$\cdot$} \tlput{$\circ$}}
 {
   \pstree[linestyle=none,levelsep=2ex]{\Tfan[fansize=10ex]}{\TR[name=C11]{$\Delta_1$}}
 }
 \pstree{\TR[name=C2]{$\cdot$} \trput{$\circ$}}
 {
   \pstree[linestyle=none,levelsep=2ex]{\Tfan[fansize=10ex]}{\TR[name=C21]{$\Delta_2$}}
 }
}
\psset{nodesep=4pt,arrows=->,linestyle=dotted,labelsep=0pt}
\ncarc[arcangle=45]{C1}{R} \naput{$\PDia$}
\ncarc[arcangle=45]{R}{C1} \naput{$\FDia$}
\ncarc[arcangle=45]{R}{C2} \naput{$\FDia$}
\ncarc[arcangle=45]{C2}{R} \naput{$\PDia$}
}
\qquad
\raisebox{-10ex}{$\Longrightarrow$}
\qquad
\mbox{
\pstree[levelsep=10ex,treesep=8ex,nodesep=2pt,labelsep=0pt]
{\TR[name=R]{$\cdot$}}
{
 \pstree{\TR[name=C1]{$\cdot$} \tlput{$\circ$}}
 {
   \pstree[linestyle=none,levelsep=2ex]{\Tfan[fansize=10ex]}{\TR[name=C11]{$\Delta_1$}}
   \pstree[linestyle=none,levelsep=2ex]{\Tfan[fansize=10ex]}{\TR[name=C12]{$\Delta_2$}}
 }
}
\psset{nodesep=4pt,arrows=->,linestyle=dotted,labelsep=0pt}
\ncarc[arcangle=45]{C1}{R} \naput{$\PDia$}
\ncarc[arcangle=45]{R}{C1} \naput{$\FDia$}
}
$$
\caption{Preservation of propagation paths in residuation and medial rules}
\label{fig:prop-res}
\end{figure}

\begin{lemma}
\label{lm:DKtP-structural-rules}
The rules $\mathit{dw, dgc, rp}$ and $\mathit{rf}$ are height-preserving
admissible in $\DKtExt{P},$ for any set of path axioms $\mbf P.$
\end{lemma}
\begin{proof}
As height-preserving admissibility of these rules have been 
proved for $\DKt$, the new cases are 
those that interact with the propagation rules in $Prop(\mbf P)$.
That is, we need to prove these for the cases where the derivation
of the premise of the rules, say $\Pi$, ends with a propagation rule: 
$$
\infer[ .]
{\Sigma[\udia A]_i[\emptyset]_j}
{\deduce{\Sigma[\udia A]_i[A]_j}{\Pi_1}}
$$
Since the propagation rule only requires the existence of a path
between node $i$ and $j$, it is sufficent to show that
a path still exists between those nodes in the modified structure. 
This is trivial for weakening. 
For the residuation rule $rp$ (the case
with $rf$ is similar), suppose that $\Pi$ is a derivation 
of the premise of $rp$, i.e., 
$$\Sigma[\udia A]_i[\emptyset]_j = \Gamma, \pseq{\Delta}$$
for some $\Gamma$ and $\Delta$. We need to show that there exists a derivation
$\Pi'$ of $\fseq{\Gamma}, \Delta.$
It is enough to show that the propagation graph
of $\fseq{\Gamma}, \Delta$ is identical to the propagation graph of
$\Gamma, \pseq \Delta$; hence any path that exists in the latter
also exists in the former, and therefore any propagation that applies
to the latter also applies to the former. 
The fact that the propagation graphs of both structures coincide
can be easily seen in the graphs in the upper row in Figure~\ref{fig:prop-res}: 
the only change caused by residuation is confined to the root of the sequent tree,
so one needs only to check that the nodes affected by these changes
are still connected with the same labelled edges. 

To prove admissibility of $\mathit{dgc}$, as with the case of $\DKt$,
we need to prove admissibility of formula contraction $\mathit{dfc}$
and the medial rules $\mathit{mf}$ and $\mathit{mp}$, as in Lemma~\ref{lm:medial}.

Admissibility of $\mathit{dfc}$ can be proved by a simple induction
on the height of derivation. 
We show here a proof of admissibility of $\mathit{mf}$; admissibility of
$\mathit{mp}$ can be proved similarly.
So suppose $\vdash_{\DKtExt{P}} \Pi : \Sigma[\fseq {\Delta_1}, \fseq{\Delta_2}]$.
We need to show that there exists $\Pi'$ such that 
$\vdash_{\DKtExt{P}} \Pi' : \Sigma[\fseq{\Delta_1,\Delta_2}]$
and $|\Pi'| = |\Pi|.$
The proof in this case is similar to the proof of the admissibility
of residuation: one shows that the modified structure still preserves
the existence of a path between two nodes where propagation happens.
Since the differences between $tree(\fseq{\Delta_1}, \fseq {\Delta_2})$
and $tree(\fseq{\Delta_1,\Delta_2})$ are confined to the top three
nodes in the trees (see the graphs in the lower row of Figure~\ref{fig:prop-res}), 
we need only to show that labelled edges between the top three nodes in the propagation
graph for $\fseq{\Delta_1}, \fseq {\Delta_2}$ 
are preserved in their corresponding nodes in 
the propagation graph for $\fseq{\Delta_1,\Delta_2}.$
This is shown in the graphs in the lower row in Figure~\ref{fig:prop-res}.
\end{proof}

\begin{lemma}
\label{lm:path-structural-rules}
Let $\mbf P$ be a set of path axioms. 
Every structural rule  in $\rho(\mbf P)$ is
admissible in $\DKtExt{P}$.
\end{lemma}
\begin{proof}
Let $\udia_1 \cdots \udia_k X \impl \udia X$ be an axiom in $\mbf P$
and let $\rho$ be its corresponding structural rule:
$$
\infer[\rho]
{\Gamma, \star_1 \{ \cdots \star_k \{ \Delta \} \cdots \} }
{\Gamma, \star \{ \Delta \} }
$$
Let $\Pi$ be a $\DKtExt{P}$-derivation  of $\Gamma, \star \{ \Delta \}$.
We show by induction on the height of $\Pi$ that there exists
a $\DKtExt{P}$-derivation $\Pi'$ of 
$\Gamma, \star_1 \{ \cdots \star_k \{ \Delta \} \cdots \}.$
Let $n_1$ denote the root node of the tree $tree(\Gamma, \star \{\Delta\})$
and let $n_2$ denote its child that is the root of its subtree $\Delta.$
So graphically, the nested sequent $\Gamma, \star \{\Delta\}$ can be represented
schematically as the tree on the left in Figure~\ref{fig:prop-axiom}. 
The tree for $\Gamma, \star_1 \{ \cdots \star_k \{ \Delta \} \cdots \}$
replaces the node $n_2$ in $tree(\Gamma,\star \{\Delta\})$
with $k$ new nodes. As $k$ could be $0$, node $n_1$ 
and node $n_2$ could possibly be identified in the conclusion of the rule $\rho.$ 
The only interesting cases are when $\Pi$ ends with a propagation
rule that propagates a $\udia A$ formula across node $n_1$ to $n_2$ or
the reverse. 
So suppose $\rho$ propagates a $\udia A$ formula along the following path:
$$
\pi_1 ; n_1 ; \udia ; n_2 ; \pi_2.
$$
This means that $\seqdia{\pi_1} \udia \seqdia{\pi_2} X \impl \udia X$ is
a member of $(\mbf P \cup I(\mbf P))^*.$ Since the set
$(\mbf P \cup I(\mbf P))^*$ is closed under axiom
composition, we also have that 
$\seqdia{\pi_1} \udia_1 \cdots \udia_k \seqdia{\pi_2} X \impl \udia X.$
The latter implies that the following
$$
\pi_1 ; n_1 ; \udia_1 ; \cdots ; \udia_k ; n_2 ; \pi_2
$$
is a path in the propagation graph of 
$\Gamma, \star_1 \{ \cdots \star_k \{ \Delta \} \cdots \}$,
so the propagation of $\udia A$ that applies to $\Gamma,\star \{\Delta\}$
can also be applied to $\Gamma, \star_1 \{ \cdots \star_k \{ \Delta \} \cdots \}.$
The other case where the propagation passes from $n_2$ to $n_1$ can
be proved symmetrically, since the set $(\mbf P \cup I(\mbf P))^*$
is closed under residuation. 
This is represented graphically in Figure~\ref{fig:prop-axiom}.
\begin{figure}[t]
$$
\mbox{
\pstree[levelsep=10ex,treesep=13ex,nodesep=2pt,labelsep=0pt]
{\TR[name=R]{$n_1$}}
{
 \pstree[linestyle=none,levelsep=2ex,arrows=-]{\Tfan[fansize=10ex]}{\TR[name=C2]{$\Gamma$}}
 \pstree{\TR[name=C1]{$n_2$} \trput{$\star$}}
        {\pstree[linestyle=none,levelsep=2ex]{\Tfan[fansize=10ex]}{\TR[name=C3]{$\Delta$}}}
}
\psset{nodesep=4pt,arrows=->,linestyle=dotted,labelsep=0pt}
\ncarc[arcangle=45]{R}{C1} \naput{$\udia$}
\ncarc[arcangle=45]{C1}{R} \naput{$\udia^{-1}$}
}
\qquad
\raisebox{-10ex}{$\Longrightarrow$}
\qquad
\mbox{
\pstree[levelsep=10ex,treesep=13ex,nodesep=2pt,labelsep=0pt]
{\TR[name=R]{$n_1$}}
{
 \pstree[linestyle=none,levelsep=2ex]{\Tfan[fansize=10ex]}{\TR[name=C2]{$\Gamma$}}
 \pstree[levelsep=5ex,linestyle=dotted]
 {\TR[name=D1]{$\cdot$} \trput{$\star_1$}}
 {
   \pstree[levelsep=10ex,linestyle=solid]{\TR[name=D2]{$\cdot$}}
   {
     \pstree{\TR[name=C1]{$n_2$} \trput{$\star_k$}}
            {\pstree[linestyle=none,levelsep=2ex]{\Tfan[fansize=10ex]}{\TR[name=C3]{$\Delta$}}}
   }
 }
}
\psset{nodesep=4pt,arrows=->,linestyle=dotted,labelsep=0pt}
\ncarc[arcangle=45]{R}{D1} \naput{$\udia_1$}
\ncarc[arcangle=45]{D1}{R} \naput{$\udia_1^{-1}$}
\ncarc[arcangle=50]{D2}{C1} \naput{$\udia_k$}
\ncarc[arcangle=50]{C1}{D2} \naput{$\udia_k^{-1}$}
}
$$
\caption{Preservation of propagation paths in structural rules for path axioms.}
\label{fig:prop-axiom}
\end{figure}
\end{proof}

\begin{theorem}
For any set of path axioms $\mbf P$ and any nested sequent $\Gamma$,
$\vdash_{\SKtExt{P}} \Gamma$ if and only if $\vdash_{\DKtExt{P}} \Gamma.$ 
\end{theorem}
\begin{proof}
This follows from Lemma~\ref{lm:DKtP-SKtP}, Lemma~\ref{lm:DKtP-structural-rules} 
and Lemma~\ref{lm:path-structural-rules}. 
\end{proof}

As with all the other extensions of $\DKt$ so far, the separation
property also holds for $\DKtExt{P}$. 
Let $\DKP$ denote the purely modal fragment of $\DKtExt{P}.$
Below we denote with $K + \mbf P$ the modal logic $K$ extended with
the axioms $\mbf P.$

\begin{theorem}[Separation]
For any set of path axioms $\mbf P$ and any modal formula $A$,
$\vdash_{\DKP} A$ if and only $A$ is a theorem of $K + \mbf P.$
\end{theorem}

\subsection{Computing the applicability of propagation rules}

Since the propagation rules of $\DKtExt{P}$ allow propagation of a formula
to a node at an arbitrary distance from the original node, depending on
the set of axioms adopted, applications of these rules are not simple
pattern matching like the local propagation rules we encountered in Section~\ref{sec:ext}.
A major obstacle in proof search for $\DKtExt{P}$ is to decide, given
a nested sequent $\Sigma[\udia A]_i[\emptyset]_j$, where $i$ and $j$ denote two
nodes in the tree of the sequent, whether the subformula $A$ 
of the occurrence of $\udia A$ at node $i$ can be propagated to
node $j.$ There are two main problems in checking whether a propagation rule
is applicable: 
\begin{enumerate}[$\bullet$]
\item there can be infinitely many paths between $i$ and $j$, and
\item there can be infinitely many combinations of axioms 
of $\mbf P$ (and its inverted versions).
\end{enumerate}
In this section we show that the decision problem of whether a
propagation rule is applicable to a nested sequent is decidable.
The main idea here is to view path axioms as representing
a context-free grammar, and the propagation graph of a nested
sequent as a finite state automaton. The problem of checking
whether a propagation rule is applicable to two nodes 
of a nested sequent is then reduced to checking whether
the intersection of a context-free grammar and a regular
language is non-empty, which is known to be decidable~\cite{ginsburg}.

Let $F$ and $P$ be two non-terminals (denoting `future' and `past'
respectively) in a context-free grammar. 
Define a function $C$ assigning diamond operators to 
either $F$ or $P$ as follows: 
$$
C(\FDia) = F \qquad C(\PDia) = P.
$$
Each path axiom $\udia_1 \cdots \udia_n X \impl \udia_{n+1} X$
defines a production rule as follows:
$$
C(\udia_{n+1}) \prod C(\udia_1) \cdots C(\udia_n).
$$
If $A$ is a path axiom, we write $G(A)$ to denote its associated
production rule defined as above. 
For example, the axiom $\FDia \PDia \FDia X \impl \FDia X$
defines the production rule $F \prod FPF.$

We recall that a context-free grammar is defined by 
a tuple $(N, T, Pr, S)$ of a set of non-terminal symbols $N$,
a set of terminal symbol $T$, a set of production rules $Pr$,
and a start symbol $S \in N.$ We shall write $F \prod^* s$ to
denote a derivation of the sequence $s$ of symbols from the 
symbol $F$.  
\begin{definition}
\label{def:cfg}
Let $\mbf P$ be a finite set of path axioms. Define two context-free
grammars generated from $\mbf P$ as follows:
\begin{enumerate}
\item Let $L_\FDia(\mbf P)$ be the grammar $(\{F,P\}, \{\FDia, \PDia\}, Pr, F)$ 
where $Pr$ is the smallest set of production rules such that:
\begin{enumerate}[(a)]
\item $F \prod \FDia$ and $P \prod \PDia$ are in $Pr$.
\item For each axiom in $A \in \mbf P \cup I(\mbf P)$, $G(A) \in Pr.$

\end{enumerate}
\item Let $L_\PDia(\mbf P)$ be the same grammar as $L_\FDia(\mbf P)$
except that the start symbol is $P$ instead of $F.$

\end{enumerate}
\end{definition}

The following lemma follows immediately from Definition~\ref{def:cfg}.
\begin{lemma}
\label{lm:cfg}
Let $\mbf P$ be a finite set of path axioms. 
Then $\udia_1 \cdots \udia_n X \impl \FDia X \in (\mbf P \cup I(\mbf P))^*$
if and only if $\udia_1 \cdots \udia_n \in L_\FDia(\mbf P).$
Similarly, $\udia_1 \cdots \udia_n X \impl \PDia X \in (\mbf P \cup I(\mbf P))^*$
if and only if $\udia_1 \cdots \udia_n \in L_\PDia(\mbf P).$
\end{lemma}

Note that the propagation graph of a nested sequent can
be seen as essentially a finite state automaton, minus the 
initial and final states. 

\begin{definition}
Let $\Gamma$ be a nested sequent, and let $n_1$ and $n_2$ be two nodes
in $tree(\Gamma).$
The {\em $(n_1,n_2)$-path automaton} of $\Gamma$, written $Path(\Gamma,n_1,n_2)$, is 
the directed graph $PG(\Gamma)$ with starting state $n_1$ and final state $n_2.$
\end{definition}

\begin{lemma}
\label{lm:fsa}
Let $\Gamma$ be a nested sequent and let $n_1$ and $n_2$ be two nodes
in $tree(\Gamma)$. Then for every $\pi$, $\pi$ is a path from 
$n_1$ to $n_2$ if and only if $\seqdia \pi \in Path(\Gamma,n_1,n_2).$
\end{lemma}

\begin{theorem}
\label{thm:cfg-fsa}
Let $\mbf P$ be a finite set of path axioms and  
let $\Gamma$ be $\Sigma[\udia A]_i[\emptyset]_j.$
Then a formula occurrence $\FDia A$ at node $i$ can
be propagated to node $j$ in the proof system $\DKtExt{P}$ if and only if 
$L_\FDia(\mbf P) \cap Path(\Gamma, i, j) \not = \emptyset.$
Similarly, a formula occurrence $\PDia A$ at node  $i$ can
be propagated to node $j$ in the proof system $\DKtExt{P}$ if and only if 
$L_\PDia(\mbf P) \cap Path(\Gamma, i, j) \not = \emptyset.$
\end{theorem}
\begin{proof}
Straightforward from Lemma~\ref{lm:cfg} and Lemma~\ref{lm:fsa}. 
\end{proof}

\begin{theorem}
\label{thm:prop-decidable}
Let $\mbf P$ be a finite set of path axioms.
Let $\Gamma$ be a nested sequent. The problem of checking whether there is 
a propagation rule in $\DKtExt{P}$ 
that is (bottom-up) applicable to $\Gamma$ is decidable.
Moreover, assuming $\mbf P$ is fixed, the complexity of the decision problem is 
PTIME in the size of $\Gamma.$
\end{theorem}
\begin{proof}
By Theorem~\ref{thm:cfg-fsa}, this decision problem reduces
to the problem of checking emptiness of the intersection
of a regular language and a context-free language, which
is itself a context-free language (see \cite{ginsburg}, Chapter 3).
Let $\mathcal{A}$ be the finite state automaton encoding paths in $\Gamma$
and let $n$ be its size. Let $\mathcal{G}$ be the context free grammar
generated from the axiom $\mbf P$ (i.e., it is either $L_\FDia(\mbf P)$ or $L_\PDia(\mbf P)$). 
In \cite{ginsburg}, the intersection of ${\mathcal A}$ and ${\mathcal G}$ is 
done by constructing another context-free grammar ${\mathcal G}'$. More specifically,
for each production rule of ${\mathcal G}$, say 
$V \prod \alpha_1 \alpha_2 \cdots \alpha_m$, where $V$ is a non-terminal
of ${\mathcal G}$ and $\alpha_i$ is either a terminal or a non-terminal of ${\mathcal G}$,
one constructs $n^{m+1}$ production rules of the same length for ${\mathcal G}'.$
So the size of ${\mathcal G'}$ is bounded by $O(l \times k \times n^{k+1})$ where $l$ is the number
of production rules in ${\mathcal G}$ and $k$ is the maximum
length of the production rules of ${\mathcal G}.$ Since the construction of each production
rule of ${\mathcal G}'$ from a production rule of ${\mathcal G}$ length $m$ takes $O(m)$-time, 
the time complexity of the construction of ${\mathcal G}'$ is also bounded by $O(l \times k \times n^{k+1}).$ 
See \cite{ginsburg} for the details of the construction of ${\mathcal G}'$.
If we assume that $\mbf P$ is fixed, then obviously $l$ and $k$ are constants,
the grammar ${\mathcal G}'$ is computable in PTIME in the size of $\Gamma$, and
its size is also polynomial in the size of $\Gamma.$
Since emptiness checking of a context-free language is 
decidable in PTIME (see e.g., \cite{Papadimitriou}),
it follows that the problem of checking the applicability of
the propagation rules is also decidable in PTIME.
\end{proof}

In some cases, the propagation rules for a given set of axioms
can be characterised by simple regular expressions. We give some examples below. 
In the following, we shall use the symbols $+$ and $*$ to denote
the union operation and the Kleene-star operations 
on regular languages. We shall be concerned only with regular
languages generated by the alphabets $\{\FDia, \PDia\}.$

\begin{example}
{\em Transitivity.}
Consider the case where ${\mbf P} = \{ \FDia \FDia X \impl \FDia X\}.$
It is easy to see that in this case, we have 
$
L_\FDia(\mbf P) = \FDia \FDia^*
$
and
$
L_\PDia(\mbf P) = \PDia \PDia^*. 
$
If one adds the axiom of reflexivity, then we get the logic $KtS4$
and the propagation paths are characterised by 
$
L_\FDia(\mbf P) = \FDia^*
$
and
$
L_\PDia(\mbf P) = \PDia^*. 
$
In other words, the propagation rules for $KtS4$ are characterised by movements
along paths of diamonds of arbitrary length and of the same color. 
\end{example}

\begin{example}
\label{ex:euclid}
{\em Euclideanness.}
Consider the case of where ${\mbf P} = \{\PDia \FDia X \impl \FDia X \}.$
Note that the inverted version of the (primitive form) of the axiom 5 is
$\PDia \FDia X \impl \PDia X.$
We claim that the paths allowed by ${\mbf P}$ can be characterised as follows:
$$
L_\FDia(\mbf P) = \FDia + (\PDia (\PDia + \FDia)^* \FDia)
\qquad
L_\PDia(\mbf P) = \PDia + (\PDia (\PDia + \FDia)^* \FDia).
$$
We prove this claim for the characterisation of $L_\FDia$, the other
case is similar. 
First, we show that $L_\FDia(\mbf P) \subseteq \FDia + (\PDia (\PDia + \FDia)^* \FDia).$
By definition, the production rules of $L_\FDia(\mbf P)$ are
$$
F \prod PF, ~ P \prod PF, ~ F \prod \FDia, ~ \mbox{ and } ~ P \prod \PDia. 
$$
It is clear that members of $L_\FDia(\mbf P)$ are either of the form $\FDia$ 
or $\PDia s \FDia$, for some sequence of diamonds $s$. 
But obviously, $s \in (\PDia + \FDia)^*$, so we indeed have $L_\FDia(\mbf P) \subseteq
\FDia + (\PDia (\PDia + \FDia)^* \FDia).$
For the other direction, suppose that $s \in \FDia + (\PDia (\PDia + \FDia)^* \FDia).$
We show by induction on the length of $s$ that $s \in L_\FDia(\mbf P).$
The case where $s = \FDia$ is trivial. So suppose $s = \PDia s' \FDia$ for some
$s' \in (\PDia + \FDia)^*.$ The case where $s'$ is the empty string is trivial; 
there remain two cases to consider:
\begin{enumerate}[$\bullet$]
\item $s' = \FDia t$ for some $t.$ By the induction hypothesis, we have
that $\PDia t \FDia \in L_\FDia(\mbf P).$ Note that the first $\PDia$
in this sequence can only be a result of the production rule $P \prod \PDia$,
so we have that 
$F \prod^* P t \FDia \prod \PDia t \FDia.$
Now, the sequence $s$ is then generated as follows:
$$
F \prod^* P t \FDia \prod PF t \FDia \prod P\FDia t \FDia \prod \PDia \FDia t \FDia = s. 
$$
\item $s' = \PDia t$ for some $t.$
By the induction hypothesis, we have $\PDia t \FDia \in L_\FDia(\mbf P)$, that is, 
we have 
$
F \prod^* \PDia t \FDia.
$
The sequence $s$ is then derived as follows:
$$
F \prod PF \prod^* P \PDia t \FDia \prod \PDia \PDia t \FDia = s.
$$
\end{enumerate}

The above characterisation of the propagation rules for axiom 5 basically says
that a formula such as $\FDia A$ can be propagated along paths of
the following form: it is either $\FDia$, or it must start with $\PDia$,
followed by any path of arbitrary length, and end with $\FDia.$ 
Using this characterisation, one can replace the generic propagation rule for 
$\FDia$-formulae in $Prop(\mbf P)$ (see Definition~\ref{def:prop-rules}), 
with more specific rules in the following (in addition to the rules $\FDia_1$ 
and $\FDia_2$ in Figure~\ref{fig:DKt}): 
$$
\infer[\mathit{p5}_a]
{\Sigma[\fseq{\FDia A, \Delta}][\fseq{\Gamma}]}
{
\Sigma[\fseq{\FDia A, \Delta}][\fseq{\Gamma, A}]
}
\quad
\infer[\mathit{p5}_b]
{\Sigma[\fseq{\FDia A, \Delta}][\pseq{\Gamma}]}
{\Sigma[\fseq{\FDia A, \Delta}][\pseq{\Gamma}, A]}
$$
$$
\infer[\mathit{p5}_c]
{\Sigma[\FDia A, \pseq{\Delta}][\fseq{\Gamma}]}
{\Sigma[\FDia A, \pseq{\Delta}][\fseq{\Gamma, A}]}
\qquad
\infer[\mathit{p5}_d]
{\Sigma[\FDia A, \pseq{\Delta}][\pseq{\Gamma}]}
{\Sigma[\FDia A, \pseq{\Delta}][\pseq{\Gamma}, A]}
$$
\end{example}

In the purely modal setting, the propagation rule $\mathit{p5}_a$ 
in Example~\ref{ex:euclid} above is similar to that considered by 
Br\"unnler~\cite{Brunnler09tableaux}:\footnote{This is not the exact form of the rule given in \cite{Brunnler09tableaux}, but it 
describes the same rule.}
$$
\infer[\stackrel{\diamond}{5}]
{\Sigma[\fseq{\Gamma, \FDia A}][\emptyset]}
{\Sigma[\fseq{\Gamma}][\FDia A]}
$$
But notice that, unlike our propagation rules, Br\"unnler's rule allows propagation
of $\FDia A$ without introducing the connective $\FDia.$

\begin{example}
\label{ex:S5}
{\em S5.}
If one adds the axiom of reflexivity $X \impl \FDia X$ to the set ${\mbf P}$ in the previous example,
one gets the logic $S5$. In this case, the propagation rules admit a very simple characterisation:
the formula $\FDia A$ (likewise, $\PDia A$) in a node $u$ in a tree of sequents
can be propagated to {\em any} node in the the tree, i.e., we have 
$L_\FDia(\mbf P) = L_\PDia(\mbf P) = (\FDia + \PDia)^*.$ 
\end{example}

\section{Proof search in \texorpdfstring{$\DKt$}{DKt}}
\label{sec:search}

\begin{figure}[t]
Function Prove (Sequent $\Xi$) : Bool
\begin{enumerate}
	\item Let $T = tree(\Xi)$
	\item If the $id$ rule is applicable to any node in $T$, return $True$
	\item\label{step:saturate} Else if there is some node $\Theta \in T$ that is not saturated
	\begin{enumerate}
		\item If $A \lor B \in \Theta$ and $A \notin \Theta$
                  or $B \notin \Theta$ then let $\Xi_1$ be the premise
                  of the $\lor$ rule applied to $A \lor B \in
                  \Theta$. Return $Prove(\Xi_1)$.
		\item If $A \land B \in \Theta$ and $A \notin \Theta$
                  and $B \notin \Theta$ then let $\Xi_1$ and $\Xi_2$
                  be the premises of the $\land$ rule applied to $A
                  \land B \in \Theta$. Return $True$ iff $Prove(\Xi_1)
                  = True$ and $Prove(\Xi_2) = True$.
	\end{enumerate}
	\item\label{step:append} Else if there is some node $\Theta
          \in T$ that is not realised, i.e. some $B = \FBox A$ ($B =
          \PBox A$) is not realised
	\begin{enumerate}
		\item Let $\Xi_1$ be the premise of the $\FBox$
                  ($\PBox$) rule applied to $B \in \Theta$. Return
                  $Prove(\Xi_1)$.
	\end{enumerate}
	\item\label{step:propagate} Else if there is some node
          $\Theta$ that is not propagated
		\begin{enumerate}
			\item Let $\rho$ be the rule corresponding to
                          the requirement of
                          Definition~\ref{def:prop-fails} that is not
                          met, and let $\Xi_1$ be the premise of
                          $\rho$. Return $Prove(\Xi_1)$.
		\end{enumerate}
	\item Else return $False$
\end{enumerate}
\caption{Proof search strategy for $\DKt$}
\label{fig:strategy}
\end{figure}

We now present a preliminary result in proof search for $\DKt.$
This section is meant to serve as a preview of our planned future work 
in designing more general proof search strategies 
for a wide range of deep inference calculi discussed in the
previous section. 

We shall be working with the tree representation of
nested sequents as discussed in the previous section.
However, since contraction is admissible in $\DKt$ and its extensions
discussed so far, we shall consider a node as a set rather than
a multiset. While traditional tableaux methods operate on a single node
at a time, our proof search strategies will consider the whole
tree. Our proof search strategy is based on a saturation procedure
familiar from the tableaux setting. 
In the following, given a tree $T$ of sequents and a node $u$ in $T$, 
we denote with $S(u)$ the set of formulae at node $u$. 

\begin{definition}
A set of formulae $\Theta$ is {\em saturated} iff it satisfies:
\begin{enumerate}
\item If $A \lor B \in \Theta$ then $A \in \Theta$ and $B \in \Theta$.
\item If $A \land B \in \Theta$ then $A \in \Theta$ or $B \in \Theta$.
\item For every propositional variable $p$, $p \in \Theta$ implies
$\neg p \not \in \Theta$, and $\neg p \in \Theta$ implies $p \not \in \Theta.$
\end{enumerate}
A node $u$ in a tree $T$ is {\em saturated} iff $S(u)$ is saturated.
\end{definition}

\begin{definition}
  Given a tree $T$ and a node $u$ in $T$, a formula $\FBox A \in
  S(u)$ $(\PBox A \in S(u))$ is {\em realised} iff there exists a
  $\circ$-child $(\bullet$-child$\,)$ $v$ of $u$ in $T$ with 
  $A \in S(u)$.
\end{definition}

\begin{definition}\label{def:prop-fails}
Given a tree $T$ and a node $u$ in $T$, we say $u$ is
propagated iff:
\begin{description}
	\item[$\FDia_1$:] for every $\FDia A \in S(u)$ and for every
          $\circ$-child $v$ of $u$, we have $A \in S(v)$;
	\item[$\PDia_1$:] for every $\PDia A \in S(u)$ and for every
          $\bullet$-child $v$ of $u$, we have $A \in S(v)$;
	\item[$\FDia_2$:] for every $\bullet$-child $v$ of
          $u$ and for every $\FDia A \in S(v)$, we have $A \in S(u)$;
	\item[$\PDia_2$:] for every $\circ$-child $v$ of $u$
          and for every $\PDia A \in S(v)$, we have $A \in S(u).$
\end{description}
\end{definition}
Figure~\ref{fig:strategy} gives a proof search strategy for
$\DKt$. The application of a rule deep inside a sequent can be viewed
as focusing on a particular node of the tree. The rules of $\DKt$ can
then be viewed as operations on the tree encoded in the sequent. In
particular, Step~\ref{step:saturate} saturates a node locally,
Step~\ref{step:append} appends new nodes to the tree, and
Step~\ref{step:propagate} moves $\FDia$ ($\PDia$) prefixed formulae
between neighbouring nodes.

The degree of a formula is the maximum number of nested modalities:
$$
\begin{array}{rcl}
deg(p) & = & 0 \\ 
deg (A \# B) & = & max(deg(A), deg(B)) \text { for } \# \in \{ \land, \lor \} \\
deg(\# A) & = & 1+ deg(A) \text { for } \# \in \{ \FBox, \FDia, \PBox, \PDia \}. \\
\end{array}
$$
The degree of a set of formulae is the maximum degree over all its
members. We write $\mathit{sf}(A)$ for the subformulae of $A$, and
define the set of subformulae of a set $\Theta$ as
$\mathit{sf}(\Theta) = \bigcup_{A \in \Theta} \mathit{sf}(A)$. For a
sequent $\Xi$ we define $\mathit{sf}(\Xi)$ as below:
$$
\begin{array}[c]{lll}
\Xi & = & \Theta, \fseq {\Gamma_1}, \cdots, \fseq {\Gamma_n}, \pseq {\Delta_1}, \cdots, \pseq {\Delta_m}
\\[0.5em]
\mathit{sf}(\Xi) &  = & \mathit{sf}(\Theta) \cup \mathit{sf}(\Gamma_1) \cup \cdots \cup
\mathit{sf}(\Gamma_n) \cup \mathit{sf}(\Delta_1) \cup \cdots \cup \mathit{sf}(\Delta_m).
\end{array}
$$

\begin{theorem}
\label{thm:termination}
Function {\rm Prove} terminates for any input sequent $\Xi$.
\end{theorem}
\begin{proof}
Let $m = |sf(\Xi)|$, $d = deg(sf(\Xi)) \leq m$ and $T = tree(\Xi)$.
The saturation process for each node in $T$ is bounded by
$m$. Therefore after at most $m$ moves at each node,
Step~\ref{step:saturate} is no longer applicable to this node.
$T$ is finitely branching, since new nodes are only created for
unrealised box formulae. Therefore after at most $m$ moves at each
node, Step~\ref{step:append} is no longer applicable to this node. The
depth of $T$ is bounded by $d$, since each node $u$ in $T$ at
distance $k$ from the root of $T$ has $degree(S(u)) \leq d - k$.
Since $\FDia$- and $\PDia$-prefixed formulae are only propagated to
nodes that do not already contain these formulae, after at most $m$
propagation moves into each node, Step~\ref{step:propagate} is no
longer applicable to this node.
\end{proof}

We now show that the procedure Prove is sound and complete
with respect to $\DKt.$ A typical semantic completeness proof would 
construct a countermodel from a failed proof search. In the following
proofs, however, we shall use purely proof theoretic arguments without
reference to semantics, unlike, say, completeness proof for a similar
procedure for modal logics in \cite{brunnler2006}. 

\begin{lemma}
\label{lm:DKt-saturated}
Let $\Xi$ be a sequent such that each node in 
$tree(\Xi)$ is saturated, realised, and propagated.
Then $\Xi$ is not derivable in $\DKt$.
\end{lemma}
\begin{proof}
We prove this by contradiction: Assume that $\Xi$ has a derivation,
therefore it also has a shortest derivation, say $\Pi$. We show that one
can construct an even shorter derivation, hence contradicting the assumption.
This is done by exploiting the fact that $tree(\Xi)$ is saturated,
realised and propagated, and Lemma~\ref{lm:medial} (essentially,
height-preserving admissibility of contraction). 
We show that every attempt to apply a rule to $\Xi$ will lead to
a duplication of formulae or create unnecessary 
structures (in the sense of the medial rules). 
We show here one case involving the rule $\FBox$; the others are similar. 

Suppose $\Pi$ ends with the rule $\FBox$. 
In this case we have $\Xi = \Sigma[\FBox A, \fseq {A, \Delta}]$,
for some context $\Sigma[~]$ and some sequent $\Delta$, such that
the rule $\FBox$ is applied to $\FBox A$ in the context $\Sigma[~]$. 
Note that $\fseq{A, \Delta}$ must also be in the same context since
every node of $tree(\Xi)$ is realised. 
$\Pi$ in this case takes the form:
$$
\infer[\FBox]
{\Sigma[\FBox A, \fseq {A, \Delta}]}
{
 \deduce{\Sigma[\FBox A, \fseq {A}, \fseq{A, \Delta}]}{\Pi'}
}
$$
Applying Lemma~\ref{lm:medial} to $\Pi'$, we get a derivation $\Pi_1$
of $\Sigma[\FBox A, \fseq{A, A, \Delta}]$ such that
$|\Pi_1| = |\Pi'|$, and applying the same lemma to $\Pi_1$,
we get another derivation $\Pi_2$ of $\Sigma[\FBox A, \fseq{A, \Delta}]$
with $|\Pi_2| = |\Pi_1| = |\Pi'| < |\Pi|$.

Since $\Pi$ cannot end with any of the rules of $\DKt$, this obviously
contradicts the assumption that it is a derivation in $\DKt$. It then 
follows that $\Xi$ is not derivable in $\DKt.$
\end{proof}

\begin{theorem}
Let $\Xi$ be a sequent. Then $\vdash_\DKt \Xi$ if and only if
{\rm Prove}($\Xi$) returns $True.$
\end{theorem}
\begin{proof}
Soundness of the Prove procedure is obvious since each
of Step 1 -- Step 5 are just applications of $\DKt$-rules.
By Therem~\ref{thm:termination}, Prove($\Xi$) always terminates and
returns either $\mathit{True}$ or $\mathit{False}$.
To show completeness, we show that if Prove($\Xi$) returns $False$
then $\Xi$ is not derivable in $\DKt.$ 
Note that each rule of $\DKt$ is invertible, hence Step 1 -- Step 5 
in Prove preserves provability of the original sequent. 
If Prove($\Xi$) returns false, this can only
be the case if Step 6 is reached, i.e., the systematic
bottom-up applications of the rules of $\DKt$ produce
a sequent such that every node in the tree of the sequent is saturated, 
realised, and propagated. By Lemma~\ref{lm:DKt-saturated}, such
a sequent would not be derivable, and since all other steps of Prove
preserves derivability, it follows that $\Xi$ is not derivable
either in $\DKt.$
\end{proof}

\section{Conclusion and related work}
\label{sec:related}

This work started out as an attempt to manage proof search in display
calculi, in particular, display calculi for tense logics by Kracht~\cite{kracht-power}. 
Due to the high-degree of non-determinism in display calculi, 
our approach was to first consider a restricted form of display calculi
with good properties, in particular, it should allow one to prove
cut elimination in a uniform manner as in display calculi, but also
close enough to traditional sequent calculi, so that traditional
proof search methods, e.g., those based on saturation of sequents, can be
applied. We have turned to nested sequent calculi for tense logics,
as originally studied by Kashima~\cite{kashima-cut-free-tense}, as a compromise;
nested sequents are more restricted than display sequents, 
but they still allow an important property, i.e., the display property,
to be proved. The display property is essentially what makes it possible
to prove cut elimination uniformly.
More interestingly, our re-formulation of tense logics in nested sequent
calculi allows us to observe an important connection between display
postulates and structural rules (in the shallow calculi) and deep inference
and propagation rules (in the deep inference calculi). We exploit this
connection to get rid of all structural rules, which are the main
obstacle to proof search, in the deep inference calculi.
We have shown a preliminary result in structuring proof search for $\DKt$.
In the future, we hope to extend this to other extensions of $\DKt.$
We need to emphasize that our work is first and foremost a proof theoretic
investigation of a proof search framework. Whether or not an efficient
decision procedure can be built on top of our framework is an important
question, but one which is out of the scope of the present paper. 

\paragraph{\em Related work}  
Areces and Bernardi~\cite{arecesbernardi2004} appear to be the first to have
noticed the connection between deep inference and residuation in
display logic in the context of categorial grammar, although 
they do not give an explicit proof of this correspondence.
Lamarche~\cite{lamarche} proposes an approach to eliminating display postulates
by moving to a more general theory of contexts in which reversible structural
rules like display postulates are treated as part of the algebraic definition of
contexts, and gives a cut elimination procedure for substructural logics 
defined using this more general notion of contexts.
Br{\"u}nnler~\cite{brunnler2006,Brunnler09AML,Brunnler09tableaux} and
Poggiolesi~\cite{poggiolesi2009} have given deep inference calculi for
the modal logic $K$ and some extensions. Sadrzadeh and Dyckhoff~\cite{Sadrzadeh10}
have given a syntactic cut elimination procedure for some extensions
of positive tense logic, i.e., tense logic without negation or implication.
Br\"unnler has recently shown that the deep-inference-based cut elimination technique
for $K$~\cite{brunnler2006} can be extended to prove cut elimination for 
Kashima's $\S2Kt$.\footnote{K. Br\"unnler. Personal communication.}
In his proof, a crucial step is a proof of the admissibility of a ``deep'' version of
residuation:
$$
\infer[]
{\Sigma[\Delta, \pseq{\Gamma}]}
{\Sigma[\pseq{\fseq{\Delta}, \Gamma}]}
\qquad
\infer[]
{\Sigma[\Delta, \fseq{\Gamma}]}
{\Sigma[\fseq{\pseq{\Delta},\Gamma}]}
$$
More recently, Br\"unnler and Stra{\ss}burger~\cite{Brunnler09tableaux} have shown how one
can extend, modularly, their deep inference calculus for modal logic $K$
with several standard axioms of normal modal logics. It is worth noting
that their formulation of these extensions allow for structural rules
to be present in the deep inference systems, contrary to our approach.
In our setting, modular extensions of tense logic are easily achieved
in the shallow setting. There is, however, a catch: our modularity result does not
imply modularity in the modal fragments. This is because our modularity
result relies on the display property, which in turn relies on the
presence of both modal and tense structural connectives. 

Indrzejczak~\cite{indrzejczak-multiple-sequent-tense} and
Trzesicki~\cite{trzesicki-gentzen-tense-logic} have given cut-free
sequent-like calculi for tense logic. In each such calculus there is a
rule (or rules) which allow us to ``return'' to previously seen worlds
when the rules are viewed from the perspective of counter-model
construction. However, Trzesicki's calculus has a large degree of
non-determinism and is therefore not suitable for proof search. In
contrast, our system $\DKt$ admits a simple proof search strategy and termination argument. 
Indrzejczak's calculus is suitable for proof search but
lacks a natural notion of a cut rule and cut elimination. It is also
possible to give proof calculi for many modal and tense logics using
semantic methods such as labelled deduction~\cite{negri2005} and graph
calculi~\cite{CaFa-De-CeGa.ea97}, but we prefer purely syntactic
methods since they can potentially be applied to logics with more
complicated semantics such as substructural logics.

\paragraph{\em Future work}
The immediate future work is to devise a terminating proof strategy for
each extension of $\DKt$ with path axioms. For extensions that include
transitivity, e.g., $KtS4$, one would need to perform loop checking
as in Heuerding's proof calculus for $S4$~\cite{heuerding1998} to
ensure termination. 
Although we have shown that one can compile any set of path axioms
into a complete set of propagation rules, it will be more desirable
if one can do it using only local propagation rules. Another interesting
avenue for future work is to investigate compositions of path axioms
with other axioms. For instance, a composition of path axioms with the seriality
axiom ($\FBox A \impl \FDia A$) will allow us to capture all fifteen basic
modal logics. Another problem is to find a complete set of propagation
rules for the confluence axiom $(\FDia  \FBox A \impl \FBox \FDia A).$
It is also interesting to see whether the connection between deep inference
and display postulates can be extended to calculi with more complex binary 
residuation principles like those in substructural logics~\cite{arecesbernardi2004}.
Another interesting direction is the addition of (first-order) quantifiers. An approach to this
would be to consider quantifiers as modal operators, with appropriate display
postulates, such as the ones developed in \cite{wansing98}.

%

\section*{Acknowledgement}
The authors wish to thank the anonymous referees of an earlier draft of this 
paper for their helpful comments.


\begin{thebibliography}{10}

\bibitem{aczel-handbook}
Peter Aczel.
\newblock An introduction to inductive definitions.
\newblock In J~Barwise, editor, {\em Handbook of Mathematical Logic}, pages
  739--782. North-Holland Publishing Company, 1977.

\bibitem{arecesbernardi2004}
Carlos Areces and Raffaella Bernardi.
\newblock Analyzing the core of categorial grammar.
\newblock {\em Journal of Logic, Language, and Information}, 13(2):121--137,
  2004.

\bibitem{Belnap82JPL}
Nuel Belnap.
\newblock Display logic.
\newblock {\em Journal of Philosophical Logic}, 11:375--417, 1982.

\bibitem{brunnler2006}
Kai Br{\"u}nnler.
\newblock Deep sequent systems for modal logic.
\newblock In G.~Governatori et~al, editor, {\em Advances in Modal Logic 6},
  pages 107--119. College Publications, 2006.

\bibitem{Brunnler09AML}
Kai Br{\"u}nnler.
\newblock Deep sequent systems for modal logic.
\newblock {\em Archive for Mathematical Logic}, 48(6):551--577, 2009.

\bibitem{BruHabil}
Kai Br{\"u}nnler.
\newblock Nested sequents.
\newblock {\em CoRR}, abs/1004.1845, 2010.

\bibitem{Brunnler09tableaux}
Kai Br{\"u}nnler and Lutz Stra{\ss}burger.
\newblock Modular sequent systems for modal logic.
\newblock In Giese and Waaler \cite{tableaux2009}, pages 152--166.

\bibitem{Brunnler01LPAR}
Kai Br{\"u}nnler and Alwen Tiu.
\newblock A local system for classical logic.
\newblock In {\em LPAR}, volume 2250 of {\em Lecture Notes in Computer
  Science}, pages 347--361. Springer, 2001.

\bibitem{CaFa-De-CeGa.ea97}
Marcos~A. Castilho, Luis~Fari{\~n}as del Cerro, Olivier Gasquet, and Andreas
  Herzig.
\newblock Modal tableaux with propagation rules and structural rules.
\newblock {\em Fundamenta Informaticae}, 32(3-4):281--297, 1997.

\bibitem{tableaux2009}
Martin Giese and Arild Waaler, editors.
\newblock {\em Automated Reasoning with Analytic Tableaux and Related Methods,
  18th International Conference, TABLEAUX 2009, Oslo, Norway, July 6-10, 2009.
  Proceedings}, volume 5607 of {\em Lecture Notes in Computer Science}.
  Springer, 2009.

\bibitem{ginsburg}
Seymour Ginsburg.
\newblock {\em The Mathematical Theory of Context-Free Languages}.
\newblock McGraw-Hill, Inc., New York, NY, USA, 1966.

\bibitem{Gore98IGPL}
Rajeev Gor{\'e}.
\newblock Substructural logics on display.
\newblock {\em LJIGPL}, 6(3):451--504, 1998.

\bibitem{Gore09tableaux}
Rajeev Gor{\'e}, Linda Postniece, and Alwen Tiu.
\newblock Taming displayed tense logics using nested sequents with deep
  inference.
\newblock In Giese and Waaler \cite{tableaux2009}, pages 189--204.

\bibitem{Gore07JLC}
Rajeev Gor{\'e} and Alwen Tiu.
\newblock Classical modal display logic in the calculus of structures and
  minimal cut-free deep inference calculi for s5.
\newblock {\em J. Log. Comput.}, 17(4):767--794, 2007.

\bibitem{heuerding1998}
Alain Heuerding, Michael Seyfried, and Heinrich Zimmermann.
\newblock Efficient loop-check for backward proof search in some non-classical
  propositional logics.
\newblock In {\em TABLEAUX}, volume 1071 of {\em LNCS}, pages 210--225.
  Springer, 1996.

\bibitem{indrzejczak-multiple-sequent-tense}
Andrzej Indrzejczak.
\newblock Multiple sequent calculus for tense logics.
\newblock International Conference on Temporal Logic, Leipzig 2000.
\newblock 93--104.

\bibitem{kashima-cut-free-tense}
Ryo Kashima.
\newblock Cut-free sequent calculi for some tense logics.
\newblock {\em Studia Logica}, 53:119--135, 1994.

\bibitem{kracht-power}
Marcus Kracht.
\newblock Power and weakness of the modal display calculus.
\newblock In Heinrich Wansing, editor, {\em Proof Theory of Modal Logics},
  pages 92--121. Kluwer, 1996.

\bibitem{lamarche}
Fran{\c{c}}ois Lamarche.
\newblock On the algebra of structural contexts.
\newblock Accepted at \emph{Mathematical Structures in Computer Science}, 2007.

\bibitem{Lemmon77}
Edward~John Lemmon and Dana~S. Scott.
\newblock {\em An Introduction to Modal Logic}.
\newblock Blackwell, Oxford, 1977.

\bibitem{negri2005}
Sara Negri.
\newblock Proof analysis in modal logic.
\newblock {\em JPL}, 34(5--6):507--544, 2005.

\bibitem{Papadimitriou}
Christos~H. Papadimitriou.
\newblock {\em Computational Complexity}.
\newblock Addison-Wesley Publishing Company, Inc., USA, 1994.

\bibitem{poggiolesi2009}
Francesca Poggiolesi.
\newblock The tree-hypersequent method for modal propositional logic.
\newblock {\em Trends in Logic: Towards Mathematical Philsophy}, pages 9--30,
  2009.

\bibitem{Postniece10Phd}
Linda Postniece.
\newblock {\em Proof Theory and Proof Search of Bi-Intuitionistic and Tense
  Logic}.
\newblock PhD thesis, The Australian National University, 2011. 

\bibitem{Sadrzadeh10}
Mehrnoosh Sadrzadeh and Roy Dyckhoff.
\newblock Positive logic with adjoint modalities: Proof theory, semantics and
  reasoning about information.
\newblock {\em Review of Symbolic Logic}, 3:351--373, 2010.

\bibitem{Troelstra96bpt}
Anne~S. Troelstra and Helmut Schwichtenberg.
\newblock {\em Basic Proof Theory}.
\newblock Cambridge University Press, 1996.

\bibitem{trzesicki-gentzen-tense-logic}
Kazimierz Trzesicki.
\newblock Gentzen-style axiomatization of tense logic.
\newblock {\em Bulleting of the Section of Logic}, 13(2):75--84, 1984.

\bibitem{wansing98}
Heinrich Wansing.
\newblock {\em Displaying Modal Logic}.
\newblock Kluwer Academic Publishers, 1998.

\end{thebibliography}
\end{document}